\documentclass[journal]{IEEEtran}
\usepackage{graphicx}
\usepackage{amscd,amsmath,amssymb,amsthm}
\usepackage{algorithm2e}
\usepackage{verbatim}
\usepackage{paralist}
\usepackage{subfigure}
\usepackage{multirow} % provides the command needed for spanning rows: \multirow{num_rows}{width}{contents}.
\usepackage{listings} % include source code in Latex
\usepackage{setspace}
\usepackage{booktabs}
\usepackage{url}

\newtheorem{defn}{Definition}
\newtheorem{thm}{Theorem}
\newtheorem{lem}{Lemma}
\newtheorem{rem}{Remark}

\def\A{{\cal A}} % HA
\def\B{{\cal B}} % HA
\def\C{{\cal C}} % HA
\def\D{{\cal D}} % set of discrete steps
\def\H{{\cal H}} % HA
\def\P{{\cal P}} % set of modes
\def\Q{{\cal Q}} % set of modes
\def\R{{\cal R}} % relation
\def\V{{\cal V}} % Lyapunov function
\def\U{{\cal U}} % set of trajectories
\def\X{{\cal X}} % Lyapunov function
\newcommand{\num}[1]{\relax\ifmmode \mathbb #1\else $\mathbb #1$\fi}
\newcommand{\nnnum}[1]{\relax\ifmmode
  {\mathbb #1}_{\geq 0} \else ${\mathbb #1}_{\geq 0}$
  \fi}
\newcommand{\npnum}[1]{\relax\ifmmode
  {\mathbb #1}_{\leq 0} \else ${\mathbb #1}_{\leq 0}$
  \fi}
\newcommand{\pnum}[1]{\relax\ifmmode
  {\mathbb #1}_{> 0} \else ${\mathbb #1}_{> 0}$
  \fi}
\newcommand{\nnum}[1]{\relax\ifmmode
  {\mathbb #1}_{< 0} \else ${\mathbb #1}_{< 0}$
  \fi}
\newcommand{\plnum}[1]{\relax\ifmmode
  {\mathbb #1}_{+} \else ${\mathbb #1}_{+}$
  \fi}
\newcommand{\nenum}[1]{\relax\ifmmode
  {\mathbb #1}_{-} \else ${\mathbb #1}_{-}$
  \fi}
\RestyleAlgo{ruled}

\SetKwData{Err}{error}
\SetKwFunction{Done}{done}
\SetKwData{False}{false}
\SetKwData{True}{true}
\SetKwData{Null}{null}
\SetKwData{Out}{out}
\SetKwFunction{Jump}{jump}
\SetKwData{Data}{data}
\SetKwData{RSCH}{ReachSetHistory}
\SetKwData{TRH}{TransitionHistory}
\SetKwFunction{Deterministic}{deterministic}
\SetKwFunction{Transversal}{transversal}
\SetKwFunction{Post}{Post()}
\SetKwFunction{Reached}{Reached}

\SetKwFunction{IsOA}{IsOverApproximate}
\SetKwFunction{RN}{ReachNext}
\SetKwFunction{IA}{ImageAt}
\SetKwFunction{IsES}{IsEpsilonSmall}
\SetKwFunction{IsTRS}{IsTransition}
\SetKwFunction{AtT}{AtTransition}
\SetKwFunction{IsOT}{IsOverTime}
\SetKwFunction{IsON}{IsOverTransition}
\SetKwFunction{IsD}{IsDeterministic}
\SetKwFunction{IsTRV}{IsTransversal}
\SetKwFunction{ImageFT}{ImageFromTo}
\SetKwFunction{ErrAT}{ErrorAtTransition}
\SetKwFunction{ErrLTI}{ErrorLTIImage}
\SetKwFunction{ErrOA}{ErrorOverApproximation}
\SetKwFunction{Compute}{Compute}
\SetKwFunction{Get}{get}
\SetKwFunction{Add}{add}
\SetKwFunction{Head}{head}
\SetKwFunction{Next}{next}
\SetKwFunction{Tail}{tail}
\SetKwFunction{Goto}{goto}
\SetKwFunction{Compute}{compute}
\SetKwFunction{Call}{call}
\SetKwFunction{Store}{store}
\SetKwFunction{Update}{update}
\SetKwFunction{ReachSet}{ReachSet}
\SetKwFunction{Transition}{Transition}

\ifCLASSINFOpdf
\else
\fi
\begin{document}
%
% paper title
% can use linebreaks \\ within to get better formatting as desired
\title{Bounded $\epsilon$-Reach Set Computation of a Class of Deterministic and Transversal Linear Hybrid Automata}
%
%
% author names and IEEE memberships
% note positions of commas and nonbreaking spaces ( ~ ) LaTeX will not break
% a structure at a ~ so this keeps an author's name from being broken across
% two lines.
% use \thanks{} to gain access to the first footnote area
% a separate \thanks must be used for each paragraph as LaTeX2e's \thanks
% was not built to handle multiple paragraphs
%

\author{Kyoung-Dae~Kim, Sayan~Mitra,~and~P. R.~Kumar
\thanks{Kyoung-Dae Kim and P. R. Kumar are with the Department
of Electrical and Computer Engineering at Texas A\&M University, and Sayan Mitra is with the Department
of Electrical and Computer Engineering and Coordinated Science Laboratory, University of Illinois at Urbana-Champaign, USA. 
e-mail: {\tt\small \{kdkim, prk\}@tamu.edu, mitras@illinois.edu}}
\thanks{This material is based upon work partially supported by NSF under Contracts CNS-1035378, CNS-1035340, and CCF-0939370, USARO under Contract Nos. W911NF-08-1-0238 and W-911-NF-0710287, and AFOSR under Contract FA9550-09-0121.}}% <-this % stops a space

% note the % following the last \IEEEmembership and also \thanks - 
% these prevent an unwanted space from occurring between the last author name
% and the end of the author line. i.e., if you had this:
% 
% \author{....lastname \thanks{...} \thanks{...} }
%                     ^------------^------------^----Do not want these spaces!
%
% a space would be appended to the last name and could cause every name on that
% line to be shifted left slightly. This is one of those "LaTeX things". For
% instance, "\textbf{A} \textbf{B}" will typeset as "A B" not "AB". To get
% "AB" then you have to do: "\textbf{A}\textbf{B}"
% \thanks is no different in this regard, so shield the last } of each \thanks
% that ends a line with a % and do not let a space in before the next \thanks.
% Spaces after \IEEEmembership other than the last one are OK (and needed) as
% you are supposed to have spaces between the names. For what it is worth,
% this is a minor point as most people would not even notice if the said evil
% space somehow managed to creep in.

% The paper headers
\markboth{IEEE Transactions on Automatic Control,~Vol.~, No.~, Month~Year}%
{Kim \MakeLowercase{\textit{et al.}}: Bounded $\epsilon$-Reach Set Computation of a Class of Deterministic and Transversal Linear Hybrid Automata}
% The only time the second header will appear is for the odd numbered pages
% after the title page when using the twoside option.
% 
% *** Note that you probably will NOT want to include the author's ***
% *** name in the headers of peer review papers.                   ***
% You can use \ifCLASSOPTIONpeerreview for conditional compilation here if
% you desire.

% If you want to put a publisher's ID mark on the page you can do it like
% this:
%\IEEEpubid{0000--0000/00\$00.00~\copyright~2007 IEEE}
% Remember, if you use this you must call \IEEEpubidadjcol in the second
% column for its text to clear the IEEEpubid mark.

% use for special paper notices
%\IEEEspecialpapernotice{(Invited Paper)}

% make the title area
\maketitle
%\IEEEpeerreviewmaketitle 

\begin{abstract}
%\boldmath
We define a special class of hybrid automata, called {\em Deterministic and Transversal Linear Hybrid Automata (DTLHA)}, whose continuous dynamics in each location are linear time-invariant (LTI) with a constant input, and for which every discrete transition up to a given bounded time is deterministic and, importantly, transversal.
For such a DTLHA starting from an initial state, we show that it is possible to compute an approximation of the reach set of a DTLHA over a finite time interval that is arbitrarily close to the exact reach set, called a {\em bounded} $\epsilon${\em -reach set}, through sampling and polyhedral over-approximation of sampled states. 
We propose an algorithm and an attendant architecture for the overall bounded $\epsilon$-reach set computation process.
\end{abstract}
% IEEEtran.cls defaults to using nonbold math in the Abstract.
% This preserves the distinction between vectors and scalars. However,
% if the journal you are submitting to favors bold math in the abstract,
% then you can use LaTeX's standard command \boldmath at the very start
% of the abstract to achieve this. Many IEEE journals frown on math
% in the abstract anyway.

% Note that keywords are not normally used for peerreview papers.

% For peer review papers, you can put extra information on the cover
% page as needed:
% \ifCLASSOPTIONpeerreview
% \begin{center} \bfseries EDICS Category: 3-BBND \end{center}
% \fi
%
% For peerreview papers, this IEEEtran command inserts a page break and
% creates the second title. It will be ignored for other modes.
\IEEEpeerreviewmaketitle

%%%%%%%%%%%%%%%%%%%%%%%%%%%%%%%%%%%%%%%%%%%%%%%%%%%%%%

%%%%%%%%%%%%%%%%%%%%%%%%%%%%%%%%%%%%%%%%%%%%%%%%%%
\section{Introduction} \label{sec:intro}

\IEEEPARstart{D} {ynamic} systems which exhibit both continuous state evolution and discrete state transitions can typically be modeled as {\em hybrid automata (HA)} (\cite{henzinger:96, lynch:03}).
Computing the reach set of a hybrid automaton from a given set of initial states is a problem of fundamental importance as it is related to safety verification and automated controller synthesis.
Even though many systems can be so modeled, it is in general undecidable to compute the exact reach set \cite{henzinger:95} except for classes of hybrid automata whose continuous dynamics are fairly simple, such as timed automata (TA) \cite{alur:94} and initialized rectangular hybrid automata (IRHA) \cite{henzinger:95}.
Neither of these automata allow the standard linear systems dynamics which is widely used for control systems. 
To broaden the class of systems that can be addressed, research in hybrid system verification in the recent years has focused on algorithms computing over-approximations of the reach set for various classes of hybrid automata (\cite{henzinger:97, frehse:08, chutinan:03, girard:05, asarin:00, kurzh:00, clarke:03, tiwari:02}).
However, even with this relaxation from exact reach set to over-approximations, it is still a challenging problem to compute an over-approximation of the reach set of hybrid automata with linear dynamics with arbitrarily small approximation error and a termination guarantee for the computation.

\subsection{Related Work}
For the computation of reach set of hybrid automata with linear dynamics, several tools and approaches have been proposed in the literature.
As an example, HyTech \cite{henzinger:97} computes the reach set of hybrid automata whose continuous dynamics are more general than those of IRHA by translating the original model into an IRHA if the model is {\em clock translatable}.
Otherwise, an over-approximate reach set is computed through an approach, called {\em linear phase-portrait approximation}, which approximates the original hybrid automaton by relaxing the continuous dynamics of the original automaton.
PHAVer \cite{frehse:08} can handle a class of systems called linear hybrid automata that have affine dynamics. 
It computes a conservative over-approximation of the reach set of such hybrid automata through on-the-fly over-approximation of the phase portrait, which is a variation of the phase-portrait approximation in \cite{henzinger:97}.
Recently, another tool, SpaceEx, has been developed based on the algorithm called LeGuernic-Girard (LGG) algorithm \cite{guernic:10} which allows the handling of hybrid automata with linear differential equations with a larger number of continuous variables compared to other approaches.

In \cite{chutinan:03}, a class of hybrid automata, called {\em polyhedral-invariant hybrid automata (PIHA)}, is defined and an algorithm is proposed to construct a finite state transition system, which is a conservative approximation of the original PIHA. 
Determining a polyhedral approximation of each sampled segment of the continuous state evolution between switching planes is the underlying fundamental technique in the algorithm that is used.
Another approach proposed in \cite{asarin:00} is also based on the idea of sampling and polyhedral over-approximation of continuous state evolution of a continuous linear dynamics. 
On the other hand, in \cite{kurzh:00} and \cite{girard:05}, ellipsoids and zonotopes are used respectively for approximating continuous state evolution.

However, while these algorithms and tools compute some over-approximation of the reach set of hybrid systems with linear dynamics, computation of an over-approximate reach set which is arbitrarily close to the exact reach set of such hybrid systems with guaranteed termination remains an open issue for further research.

\subsection{Challenges and Contributions}
In general, the key challenges in reach set computation of HA are 
\begin{inparaenum}[(i)] 
	\item to over-approximate the exact continuous flow with arbitrarily small approximation error,
	\item to determine when and where a discrete transition occurs, and
	\item to develop a reach set computation algorithm with termination guarantee. 
\end{inparaenum}
In this paper, we address the problem of computing an over-approximation of the reach set of a special class of hybrid automata, called {\em Deterministic and Transversal Linear Hybrid Automaton (DTLHA)}, starting from an initial state over a finite time interval.
We call such an over-approximate reach set as a {\em bounded $\epsilon$-reach set}.
Our approach can be related to other approaches that use sampling and polyhedral over-approximation as in \cite{chutinan:03, asarin:00}.
The main contributions of our approach are as follows: 
\begin{inparaenum}[(i)]
	\item We show that an over-approximation of the reach set of a DTLHA can be computed arbitrarily closely to the exact reach set.
	\item We also show that such computation is guaranteed to terminate under a deterministic and transversal restriction on the discrete dynamics. 
	\item Furthermore, to facilitate practical computation, we extend these theoretical results to consider the numerical calculation errors caused by finite precision calculation capabilities. 
\end{inparaenum}
Based on the theoretical results, we propose an algorithm to compute a bounded $\epsilon$-reach set of a DTLHA, as well as a software architecture that is designed to improve the flexibility and the efficiency in computing such an over-approximation.

The paper is organized as follows.
In Section \ref{sec:pre}, we introduce definitions and notations that are used throughout this paper.
In Section \ref{sec:theory}, we show that, for arbitrarily small $\epsilon > 0$, a bounded $\epsilon$-reach set of a DTLHA starting from an initial state can be computed under the assumption of infinite precision numerical calculation capabilities. 
In Section \ref{sec:cond}, we first derive a set of conditions for computation of a bounded $\epsilon$-reach set, and then extend these conditions to consider errors caused by finite precision numerical calculation capabilities.
In Section \ref{sec:design}, we propose an algorithm for a bounded $\epsilon$-reach set computation, as well as an architecture for software implementation of the proposed algorithm.
Finally, we illustrate an example of bounded $\epsilon$-reach set computation in Section \ref{sec:imp}, followed by concluding remarks in Section \ref{sec:con}.

%%%%%%%%%%%%%%%%%%%%%%%%%%%%%%%%%%%%%%%%%%%%%%%%%%
\section{Preliminaries} \label{sec:pre}

Let $\mathcal{X} \subset \mathbb{R}^n$  be a continuous state space over which a hybrid automaton is defined.
For a polyhedron $\C \subseteq \mathbb{R}^n$, we denote its interior by $\C^{\circ}$, and its boundary by $\partial \C$.
We will also use the notation $\B_r(x)$ to denote a closed ball of radius $r$ with center $x$, i.e., $\B_r(x) := \lbrace y \in \mathbb{R}^n : \Vert y - x \Vert \leq r \}$. 
The specific norm that we use in the definition of $\B_r(x)$ as well as the sequel is the $\ell_{\infty}$-norm. 
Since we are using the $\ell_{\infty}$-norm, $\B_r(x)$ is a hypercubic neighborhood of $x$.
One of the advantages of using the $\ell_{\infty}$-norm is that the induced hypercubic neighborhood is easily computed. 
More generally, a hypercube is a special case of a polyhedron, which is important since it is easy to propagate the image of this set under linear dynamics.
This is useful in Section \ref{sec:theory} when we describe our approach for bounded $\epsilon$-reach set computation.

We now describe the class of hybrid automata considered.
We assume that $\mathcal{X}$ is a closed and bounded subset of Euclidean space, and is partitioned into a collection of polyhedral regions $\mathcal{C} := \{\mathcal{C}_1, \cdots, \mathcal{C}_m \}$ such that $\C^{\circ}_i \ne \emptyset$ for each $i \in \{1, \cdots, m\}$ and

\begin{equation}\label{eq:pre:tess}
	\bigcup_{i=1}^{m} \mathcal{C}_{i} = \mathcal{X}, \quad \mathcal{C}_{i}^{\circ} \cap \mathcal{C}_{j}^{\circ} = \emptyset \quad for~i \ne j,
\end{equation}
where $m$ is the size of the partition, and each $\C_i$ is a polyhedron, called \emph{cell}. 
Two cells $\C_i$ and $\C_j$ are said to be \emph{adjacent} if the affine dimension of $\partial \mathcal{C}_i \cap \partial \mathcal{C}_j$ is $(n-1)$, or, equivalently, cells $\C_i$ and $\C_j$ intersect in an $(n-1)$-dimensional facet. 
Two cells $\C_i$ and $\C_j$ are said to be \emph{connected} if there exists a sequence of adjacent cells between $\C_i$ and $\C_j$.

\begin{defn} \label{def:lha}
An $n$-dimensional \emph{Linear Hybrid Automaton (LHA)},\footnote{In the hybrid system literature \cite{henzinger:97, alur:93} the word ``linear automaton'' has been used to denote a system where the differential equations and inequalities involved have constant right hand sides. This does not conform to the standard notion of linearity where the right hand side is allowed to be a function of state. In particular, it does not include the standard class of linear time-invariant systems that is of central interest in control systems design and analysis. We use the term ``linear'' in this latter more mathematically standard way that therefore encompasses a larger class of systems, and, more importantly, encompasses classes of switched linear systems that are of much interest.} 
is a tuple $(\mathbb{L}, Inv, A, u, \xrightarrow{G})$ satisfying the following properties: 
\begin{enumerate}[(a)]
	\item $\mathbb{L}$ is a finite set of \emph{locations} or \emph{discrete states}. The state space is $\mathbb{L} \times \mathbb{R}^n$, and an element $(l, x) \in \mathbb{L} \times \mathbb{R}^n$ is called a \emph{state}.
	\item[(b)] $Inv: \mathbb{L} \rightarrow 2^{\C}$ is a function that maps each location to a set of cells, called an \emph{invariant set} of a location, such that 
	 \begin{inparaenum}[(i)]
		\item for each $l \in \mathbb{L}$, all the cells in $Inv(l)$ are connected, 
		\item for any two locations $l, l' \in \mathbb{L}$, $Inv(l)^{\circ} \cap Inv(l')^{\circ} = \emptyset$, and 
		\item $\bigcup_{l \in \mathbb{L}} Inv(l) = \X$.
	\end{inparaenum}
	\item[(c)] $A: \mathbb{L} \rightarrow \mathbb{R}^{n \times n}$ is a function that maps each location to an $n \times n$ real-valued matrix, and 
	\item[(d)] $u: \mathbb{L} \rightarrow \mathbb{R}^n$ is a function that maps each location to an $n$-dimensional real-valued vector.
	\item[(e)] $\xrightarrow{G}: (\mathbb{R}^n, \mathbb{L}) \times (\mathbb{R}^n, \mathbb{L})$ is a binary relation which defines a \emph{discrete transition} from one state $(x_1, l_1)$ to another state $(x_2, l_2)$ such that $(x_1, l_1) \xrightarrow{G} (x_2, l_2)$ when $G$ is satisfied and $x_2$ is set to $x_1$ after a discrete transition.
\end{enumerate}
\end{defn}
In the sequel, for each $l_i \in \mathbb{L}$, we use $A_i$, $u_i$, $Inv_i$ to denote $A(l_i)$, $u(l_i)$, and $Inv(l_i)$, respectively.

An example LHA which satisfies Definition \ref{def:lha} is shown in Section \ref{sec:imp:example}.
Next, we define the behavior of LHA.
\begin{defn} \label{def:traj} %[Trajectory]
For a location $l_i \in \mathbb{L}$, a \emph{trajectory} of duration $t \in \mathbb{R}^{+}$ for an $n$-dimensional LHA $\A$ is a continuous map $\eta$ from $[0,t]$ to $\mathbb{R}^n$, such that
\begin{enumerate}[(a)]
    \item $\eta(\tau)$ satisfies the differential equation
            \begin{equation}\label{eq:pre:lti}
                \dot{\eta}(\tau) = A_i \eta(\tau) + u_i ,
            \end{equation}
    \item $\eta(\tau) \in Inv_i$ for every $\tau \in [0,t]$.
\end{enumerate}
\end{defn}

\begin{defn} \label{def:exec}%[Execution]
An \emph{execution} $\alpha$ of an LHA $\A$ from a starting state $(l_0, x_0) \in \mathbb{L} \times \mathbb{R}^n$ is defined to be the concatenation of a finite or infinite sequence of trajectories $\alpha = \eta_0 \eta_1 \eta_2 \ldots$, such that
\begin{enumerate}[(a)]
    \item $\eta_0(0) = x_0$,
    \item $\eta_{k}(0) = \eta_{k-1}(\eta_{k-1}.dur)$ for $k \ge 1$, 
\end{enumerate}
where $\eta_k$ represents a trajectory defined at some location $l \in \mathbb{L}$ and $\eta_k.dur$ denotes the duration of  $\eta_k$.
We also define $\alpha.dur := \sum_k \eta_k.dur$ where $\alpha.dur$ denotes the duration of an execution $\alpha$.
\end{defn}

We can represent an execution $\alpha$ of an LHA $\A$ from an initial condition $(l_0, x_0) \in \mathbb{L} \times \mathbb{R}^n$ for time $[0,t]$ as a continuous map $x:[0,t] \rightarrow \mathbb{R}^n$ such that
\begin{inparaenum}[(a)]
    \item $t = \alpha.dur$,
    \item $x(0) = x_0 \in Inv_0$,
    \item $x(\tau_k) = \eta_{k}(0)$, and
    \item $x(\tau) = \eta_{k-1}(\tau - \tau_{k-1})$ for $\tau \in [\tau_{k-1}, \tau_k]$,
\end{inparaenum}
where $\tau_0 = 0$, and $\tau_k = \sum_{i=0}^{k-1} \eta_i.dur$ for $k \ge 1$.
Note that $\tau_k$ for $k \ge 1$ represents the time at the $k$-th discrete transition between locations and the continuous state is not reset during discrete transitions.

\begin{defn} \label{def:trans} %[Discrete Transition]
For an execution $x(t)$ of an LHA, a discrete transition $(x_i, l_i) \xrightarrow{G} (x_j, l_j)$ occurs if $x_i = x(\tau')$ for some time $\tau'$, 
$x(\tau') \in Inv_i \cap Inv_j$ and $x(\tau') = \lim_{\tau \nearrow \tau'} x(\tau)$ where $x(\tau) \in (Inv_i)^{\circ}$ for $\tau \in (\tau'-\delta,\tau')$ for some $\delta > 0$.

\end{defn}
\begin{defn} \label{def:transversal}  %[Deterministic & Transversal Discrete Transition]
A discrete transition is called \emph{deterministic} if there is only one location $l_j \in \mathbb{L}$ to which a discrete transition state $x(\tau_k)$ can make a discrete transition from $l_i$. 
We call a discrete transition a \emph{transversal discrete transition} if there exists $\epsilon > 0$ such that
\begin{equation}\label{eq:pre:trans}
    \langle \dot{x}_{i}(\tau_k), \vec{n}_i \rangle \ge \epsilon ~~\land~~ \langle \dot{x}_{j}(\tau_k), \vec{n}_i \rangle \ge \epsilon ,
\end{equation}
where $\langle x, y \rangle$ denotes the inner product between $x$ and $y$, $\vec{n}_i$ is an outward normal vector of $\partial Inv_i$ at $x(\tau_k)$, and $\dot{x}_{i}(\tau_k) = A_i x(\tau_k) + u_i$, and $\dot{x}_{j}(\tau_k) = A_j x(\tau_k) + u_j$ are the vector fields at $x(\tau_k)$ evaluated with respect to the continuous dynamics of location $l_i$ and $l_j$, respectively. 
\end{defn}
Fig. \ref{fig:transition} illustrates a case where $x(\tau_k)$ satisfies such a deterministic and transversal discrete transition condition.
Note that if $x(\tau_k)$ satisfies a deterministic and transversal discrete transition condition, then $x(\tau_k)$ must make a discrete transition from a location $l_i$ to the other location $l_j$, and $l_j$ has to be unique. 
Furthermore, the \emph{Zeno behavior}, an infinite number of discrete transitions within a finite amount of time, does not occur if a discrete transition is a transversal discrete transition.
\begin{figure}
\begin{center}
  \includegraphics[width=7cm]{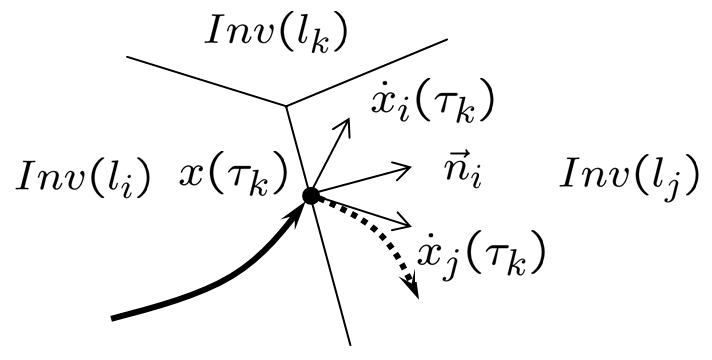} %% for 'pdftex' typeset 
\caption{A deterministic and transversal discrete transition from a location $l_i$ to a location $l_j$ occurring at $x(\tau_k) \in \partial Inv(l_i) \cap \partial Inv(l_j)$.}
\label{fig:transition}
\end{center}
\end{figure}

We now define a special class of LHA whose every discrete transition satisfies the deterministic and transversality conditions defined in Definition \ref{def:transversal} as follows:

\begin{defn} \label{pre:def:dtlha}
Given an LHA $\mathcal{A}$, a starting state $(l_0, x_0) \in  \mathbb{L} \times \X$, a time bound $T$, and a jump bound $N$, we call an LHA $\mathcal{A}$ as a \emph{Deterministic and Transversal Linear Hybrid Automaton (DTLHA)}  if all discrete transitions in the execution starting from $x_0$ up to time $t_f := \min \{T, \tau_N \}$ are deterministic and transversal, where $\tau_N$ is the time at the N-th discrete transition.
\end{defn}

Next, we define the bounded reach set of a DTLHA and its over-approximation as follows:

\begin{defn}
A continuous state in $\X$ is \emph{reachable} if there exists some time $t$ at which it is reached by some execution $x$.
\end{defn}
\begin{defn} \label{def:pre:reach}
Given a state $x_0$ and a time $t$, the \emph{bounded reach set} up to time $t$, denoted as $\R_t(x_0)$, of a DTLHA $\A$ is defined to be the set of continuous states that are reachable for some time $\tau \in [0, t]$ by some execution $x$ starting from $x_0 \in Inv_0$.
\end{defn}
\begin{defn} \label{def:pre:ereach}
Given $\epsilon > 0$, a set of continuous states $S$ is called a \emph{bounded $\epsilon$-reach set} of a DTLHA $\A$ over a time interval $[0,t]$ from an initial state $x_0$ if $\R_t(x_0) \subseteq S$ and
\begin{equation}\label{eq:pre:ereach}
	d_H(\R_t(x_0), S) \le \epsilon,
\end{equation}
where $d_H(\P, \Q)$ denotes the Hausdorff distance 
%\footnote{For two sets $X \subset \mathbb{R}^n$ and $Y \subset \mathbb{R}^n$, $d_H(X, Y) := \max \{ \sup_{x \in X} \inf_{y \in Y} d(x, y), \sup_{y \in Y} \inf_{x \in X} d(x, y) \}$ where $d(x,y) := \Vert x - y \Vert$.} 
between two sets $\P$ and $\Q$
that is defined as $d_H(\P, \Q) := \max \{ \sup_{p \in \P} \inf_{q \in \Q} d(p, q), \sup_{q \in \Q} \inf_{p \in \P} d(p, q) \}$ where $d(p,q) := \Vert p - q \Vert$.
\end{defn}

In the sequel, we use $\D_t(\P)$ to denote the set of states reached at time $t$ from a set $\P$ at time $0$.
Similarly, for the set of reached states over a time interval $[t_1, t_2)$ from $\P$, we use $\D_{[t_1, t_2)}(\P)$.
We also use $\D_t(\P, \gamma)$ to denote an over-approximation of $\D_t(\P)$ with an approximation parameter $\gamma > 0$, calling it a $\gamma$-approximation of $\D_t(\P)$ if it satisfies 
\begin{inparaenum}[(i)]
\item $\D_t(\P) \subset \D_t(\P, \gamma)$ and
\item $d_H(\D_t(\P),$ $\D_t(\P, \gamma)) \le \gamma$.
\end{inparaenum}
Note that $\D_0(\P, \gamma)$ is simply a $\gamma$-approximation of the set $\P$.

%%%%%%%%%%%%%%%%%%%%%%%%%%%%%%%%%%%%%%%%%%%%%%%%%%
\section{Bounded $\epsilon$-Reachability of a DTLHA} \label{sec:theory}

In this section, we consider the problem of a bounded $\epsilon$-reach set computation of a DTLHA starting from an initial state over a finite time interval.
More precisely, we show that, for any given $\epsilon > 0$, a DTLHA $\A$, an initial condition $(l_0, x_0) \in \mathbb{L} \times \X$, a time upper bound $T \in \mathbb{R}^+$, and a discrete transition upper bound $N \in \mathbb{N}$, it is possible to compute a bounded $\epsilon$-reach set of $\A$ over a finite time interval $[0, t_f]$ under the assumptions that the following computations can be performed exactly:
\begin{inparaenum}[(i)]
	\item  $x(t) = e^{At} x_0 + \int_0^t e^{A(t-s)} u ds$,
	\item the convex hull of a set of finite points in $\mathbb{R}^n$, and
	\item the intersection between a polyhedron and a hyperplane,
\end{inparaenum}
where $t_f$ is as defined in Definition \ref{pre:def:dtlha}, $A \in \mathbb{R}^{n \times n}$, and $u \in \mathbb{R}^n$.

%=========================================================
\subsection{Bounded $\epsilon$-Reach Set of a DTLHA at Initial Location}  \label{sec:theory:l0}

We first show how a trajectory of a DTLHA can be over-approximated through sampling and polyhedral over-approximation of each sampled state.
The basic approach for such over-approximation is shown in Fig. \ref{fig:traj}. 
It is necessary that, for a given size of over-approximation of each sampled state, a sampling period $h$ has to ensure that a trajectory $x(t)$ is contained in the computed set of polyhedra.
For a given value of $\epsilon > 0$, we now show how we can determine a sampling period $h$ which guarantees that.
\begin{equation}\label{eq:hcond}
    \max_{\tau \in [0, h]} \Vert x(t+\tau) - x(t) \Vert < \epsilon \qquad \forall x(t) \in \mathcal{X}.
\end{equation}

To determine a suitable value of $h$ which results in (\ref{eq:hcond}), we suppose $x(s) \in (Inv_i)^{\circ}$ for all $s \in [t, t+h]$ for some location $l_i \in \mathbb{L}$.
Then for a given $\Sigma_i$, $\X$, and $x(s) \in \X$, we have 

\begin{eqnarray} \nonumber
\max_{s \in [t, t+\tau]} \Vert \dot{x}(s) \Vert &=& \max_{s \in [t, t+ \tau]} \Vert A_i x(s) + u_i \Vert \\ \nonumber
&\le& \max_{s \in [t, t+\tau]} \{ \Vert A_i \Vert \Vert x(s) \Vert+ \Vert u_i \Vert \} \\ 
&\le& \Vert A_i \Vert \bar{x}+ \Vert u_i \Vert,
\end{eqnarray}
where $\bar{x} = \max_{x \in \mathcal{X}} \Vert x \Vert$. 

For a fixed $\tau \in [0, h]$, we can compute an upper bound on $\Vert x(t+\tau) - x(t) \Vert$ as follows:
\begin{eqnarray} \label{eq:h:ineq:0} \nonumber
\Vert x(t+\tau) - x(t) \Vert &\le& \int^{t+\tau}_{t} \Vert \dot{x}(s) \Vert ds \nonumber \\
&\le& \int^{t+\tau}_{t} \max_{s \in [t, t+ \tau]} \Vert \dot{x}(s) \Vert ds \nonumber \\
&\le& \int^{t+\tau}_{t} (\Vert A_i \Vert\bar{x} + \Vert u_i \Vert) ds \nonumber \\
&=& (\Vert A_i \Vert \bar{x} + \Vert u_i \Vert) \tau.
\end{eqnarray}

Maximization of both sides of (\ref{eq:h:ineq:0}) over $\tau \in [0, h]$ gives us 
\begin{eqnarray} \label{eq:h:ineq:1} \nonumber
	\max_{\tau \in [0,h]} \Vert x(t+\tau) - x(t) \Vert &\le& (\Vert A_i \Vert \bar{x} + \Vert u_i \Vert) h \\
	&\le& \max_{l_i \in \mathbb{L}} (\Vert A_i \Vert \bar{x} + \Vert u_i \Vert) h.
\end{eqnarray}

If we upper bound the right hand side by $\epsilon > 0$, then we can choose 
\begin{equation} \label{eq:h:ineq}
	h < \frac{\epsilon}{\bar{v}}.
\end{equation}
where  $\bar{v} := \max_{l_i \in \mathbb{L}} (\Vert A_i \Vert \bar{x} + \Vert u_i \Vert)$.

So, if we choose $h$ as 
\begin{equation} \label{eq:h:eq}
    h = \frac{\epsilon/2}{\bar{v}} ,
\end{equation}
then it is clear that we can ensure (\ref{eq:hcond}). 

We now show that, for a given $\epsilon > 0$, if a sampling period $h$ satisfies (\ref{eq:h:eq}), then a set constructed as a union of $\epsilon$-neighborhood of each sampled state along a trajectory is indeed a bounded $\epsilon$-reach set at an initial location.  
Moreover, such a bounded $\epsilon$-reach set contains the bounded reach set not only from the initial state but also from the $(\epsilon/2)$-neighborhood of the initial state.

\begin{figure}
\begin{center}
	\includegraphics[width=8cm]{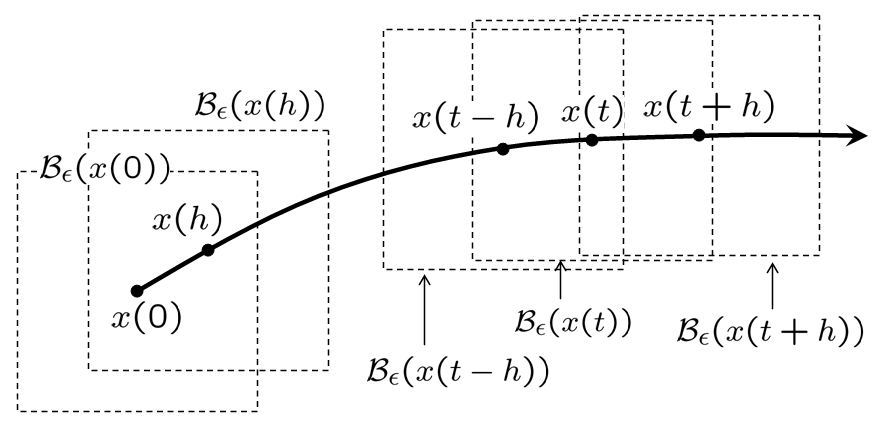}
	\caption{An over-approximation of a trajectory $x(t)$ through sampling.}\label{fig:traj}
\end{center}
\end{figure}

\begin{lem} \label{lem:theory:l0}
Given $\epsilon > 0$ and a time bound $T > 0$, a bounded $\epsilon$-reach set $\R_{t_f}(x_0, \epsilon)$ of a DTLHA $\A$ from an initial state $(x_0, l_0)$ can be determined as follows:
\begin{equation} \label{eq:lem:theory:l0}
    \R_{t_f}(x_0, \epsilon) := \bigcup_{k=0}^{m-1} \B_{\epsilon}(x(k h)) ,
\end{equation}
where $t_f := \min \{\tau_1, T \}$, $\tau_1 := \inf \{ t \in (0, T] : x(t) \not\in Inv_0 \land x(0) = x_0 \}$, $m := \lceil t_f/h \rceil$ and $h = (\epsilon/2)/\max_{l_i \in \mathbb{L}} (\Vert A_i \Vert \bar{x} + \Vert u_i \Vert)$.
Moreover, this set has two additional properties:
\begin{enumerate}[(i)]
    \item $\lim_{\epsilon \rightarrow 0} \R_{t_f}(x_0, \epsilon) = \R_{t_f}(x_0)$, and
    \item It contains an $\epsilon/2$ neighborhood of $\R_{t_f}(x_0)$, i.e.,
    \begin{equation}\nonumber
        \bigcup_{z \in \R_{t_f}(x_0)} \B_{\epsilon/2}(z) \subseteq \R_{t_f}(x_0, \epsilon).
    \end{equation}
\end{enumerate}
\end{lem}
\begin{proof}
Since $h$ satisfies (\ref{eq:h:ineq}), it is easy to see that $\R_{t_f}(x_0) \subset \R_{t_f}(x_0, \epsilon)$ from the construction of $\R_{t_f}(x_0, \epsilon)$.
Next, by the relation between $\epsilon$ and $h$ in (\ref{eq:h:eq}), it is clear that $h \rightarrow 0$ as $\epsilon \rightarrow 0$. 
This implies that $\R_{t_f}(x_0, \epsilon) \rightarrow \R_{t_f}(x_0)$ as $\epsilon \rightarrow 0$, establishing (i). 
For (ii), as noted above, (\ref{eq:h:eq}) actually chooses half the sampling period that would have sufficed to make it a bounded $\epsilon$-reach set over $[0,t_f]$. 
Hence, replacing $\epsilon$ by $\epsilon/2$ in the right hand side of (\ref{eq:lem:theory:l0}) still yields a bounded $\epsilon$-reach set. 
Thus the over stringent choice of $h$ contains not just $\R_{t_f}(x_0)$ but actually all points that are within a distance $\epsilon/2$ from it.
\end{proof}

%=========================================================
\subsection{Continuity Property of DTLHA}  \label{sec:theory:cont}

Now let us consider the problem of computing a bounded $\epsilon$-reach set of a DTLHA $\A$ not from an initial state $x_0$ but from a $\delta$-neighborhood of $x_0$.
We first show that there exists a $\delta > 0$ such that the bounded reach set of a DTLHA $\A$ from a set $\B_{\delta}(x_0)$ at an initial location $l_0$ is contained in a bounded $\epsilon$-reach set of $\A$ from $x_0$ defined in (\ref{eq:lem:theory:l0}).

\begin{lem} \label{lem:theory:cont:1}
Given $\epsilon > 0$, a time bound $T > 0$, an initial state $x_0$, and a DTLHA $\A$, there exists a $\delta > 0$ such that 
\begin{equation}
   \R_{t_f}(\B_{\delta}(x_0)) \subseteq \R_{t_f}(x_0,\epsilon) ,
\end{equation}
where $\B_{\delta}(x_0)$ is a $\delta$-neighborhood around $x_0$ and $\R_{t_f}(\B_{\delta}(x_0))$ is the bounded reach set of $\A$ from $\B_{\delta}(x_0)$ up to time $t_f$ and $t_f$ is as defined in Lemma \ref{lem:theory:l0}. 
In particular, $\R_{t_f}(\B_{\epsilon /(2C)}(x_0)) \subseteq \R_{t_f}(x_0, \epsilon)$ for an appropriate $C$.
\end{lem}
\begin{proof}
Notice that $x(t) = e^{A_0t} x_0 + \int_0^t e^{A_0(t-s)} u_0 ds$, where $A_0$ and $u_0$ define the linear dynamics in an initial location $l_0$.
If we consider two different initial states $x_0$ and $y_0$ in $\B_{\delta}(x_0)$, then their trajectories $x(t)$ and $y(t)$ satisfy $x(t) - y(t) = e^{At} (x_0 - y_0)$.
Hence $\Vert x(t) - y(t) \Vert \leq c e^{\lambda t} \Vert x_0 - y_0 \Vert$ for some positive constant $c$ and some constant $\lambda$.

Let $C := c \cdot \max_{0 \leq t \leq {t_f}}\lbrace e^{\lambda t}\rbrace$.
Then
\begin{equation} \label{eq:cont:1}
    \Vert x(t) - y(t) \Vert \leq C \Vert x_0 - y_0 \Vert \qquad \mbox{for} \quad t \in [0, {t_f}].
\end{equation}
Since $\Vert x_0 - y_0 \Vert \leq \delta$, $\Vert x(t) - y(t) \Vert \leq C \delta$ for all $t \in [0, {t_f}]$.
This implies that any initial condition $y_0$ in $\B_\delta(x_0)$ results in a $y(t)$ that lies in a $C \delta$ neighborhood of $\R_{t_f}(x_0)$ for all $t \in [0,{t_f}]$. 
In particular, from property (ii) of Lemma \ref{lem:theory:l0}, it also follows that $\R_{t_f}(\B_{\delta}(x_0)) \subseteq \R_{t_f}(x_0, 2 C \delta)$.
If we set $\delta = \epsilon /(2C)$, then it is clear that $\R_{t_f}(\B_{\delta}(x_0)) \subseteq \R_{t_f}(x_0,\epsilon)$.
\end{proof}

Next we extend the result in Lemma \ref{lem:theory:cont:1} to show that there exist a $\delta > 0$ and a $\gamma > 0$ such that an over-approximation of the bounded reach set  $\R_{t_f}(\B_{\delta}(x_0))$, denoted as $\R_{t_f}(\B_{\delta}(x_0), \gamma)$, is also contained in $\R_{t_f}(x_0,\epsilon)$ that is defined in (\ref{eq:lem:theory:l0}).

\begin{lem} \label{lem:theory:cont:2}
Given $\epsilon > 0$, a time bound $T > 0$, an initial state $x_0$, and a DTLHA $\A$, there exist $\delta > 0$ and $\gamma > 0$ such that 
\begin{equation}
	\R_{t_f}(\B_{\delta}(x_0),\gamma) \subseteq \R_{t_f}(x_0,\epsilon),
\end{equation}
where $\R_{t_f}(\B_{\delta}(x_0),\gamma)$ is a $\gamma$-approximation of $\R_{t_f}(\B_{\delta}(x_0))$, and $t_f$ is as defined in Lemma \ref{lem:theory:l0}.
In particular, $\R_{t_f}(x_0) \subseteq \R_{t_f}(\B_{\epsilon/(4C)}(x_0), \epsilon/4) \subseteq \R_{t_f}(x_0, \epsilon)$.
\end{lem}
\begin{proof}
Let $x(t;z)$ denote the solution at time $t$ of the differential equation $\dot{x}(t) = Ax(t)+u$ with initial condition $x(0) = z \in \B_\delta(x_0)$.
Now consider $w \in \R_{t_f}(\B_{\delta}(x_0), \gamma)$. 
Then, by the definition of $\R_{t_f}(\B_{\delta}(x_0))$ and $\R_{t_f}(\B_{\delta}(x_0), \gamma)$, 
\[ \Vert w - x(t;z) \Vert < \gamma\] 
for some $t \in [0,t_f]$ and $z \in \B_{\delta}(x_0)$. 
Hence
\begin{eqnarray} \nonumber
    \Vert w - x(t;x_0) \Vert &=& \Vert w - x(t;z) + x(t;z) - x(t;x_0) \Vert \\ \nonumber
    &\le& \Vert w - x(t;z) \Vert + \Vert x(t;z) - x(t;x_0) \Vert \\ \nonumber
    &\le& \gamma + \Vert x(t;z) - x(t;x_0) \Vert . \nonumber
\end{eqnarray}
From (\ref{eq:cont:1}), we know that 
\[
	\Vert x(t;z) - x(t;x_0) \Vert \le C \Vert z - x_0 \Vert \le C \delta.
\]
Hence 
\[
	\Vert w - x(t;x_0) \Vert \le \gamma + C \delta
\]
which implies that $w$ lies in a $(\gamma + C \delta)$-neighborhood of $\R_{t_f}(x_0)$.
From the property (ii) in Lemma \ref{lem:theory:l0}, if we replace $\epsilon/2$ with $(\gamma + C \delta)$, then we have $w \in \R_{t_f}(x_0, 2(\gamma + C \delta))$ which in turn implies that $\R_{t_f}(\B_{\delta}(x_0),\gamma) \subseteq \R_{t_f}(x_0, 2(\gamma + C \delta))$.
So, given $\epsilon > 0$, we can choose $\gamma = \epsilon/4$ and $\delta = \epsilon/(4C)$, and then $\R_{t_f}(\B_{\delta}(x_0),\gamma) \subseteq \R_{t_f}(x_0, \epsilon)$.
\end{proof}

%=========================================================
\subsection{Decidability of Discrete Transition Event}  \label{sec:theory:trans}

Recall that $\tau_1$ is the time $t$ when a reached state $x(t)$ of a DTLHA starting from an initial state first exits the invariant set of an initial location.
We now show that, for a given $T$, even though it is not known to be decidable to determine $\tau_1$ exactly, we can still determine the event of exit of a reached state $x(t)$ from the invariant set of an initial location if $\tau_1 < T$.

\begin{lem} \label{lem:theory:trans:exit}
Given a time bound $T > 0$, an initial condition $(l_0, x_0) \in \mathbb{L} \times \mathbb{R}^n$, and a DTLHA $\A$, if $\tau_1 < T$, then for all small enough $\delta > 0$ and for some small enough $h > 0$, $\B_{\delta}(x(nh)) \subset (Inv_0)^{c}$ for some $n \in \mathbb{N}$ satisfying $nh \le T$.
\end{lem}
\begin{proof}
Let $\vec{n}_1$ be an outward normal vector of $\partial Inv_0$ at $x(\tau_1)$.
Since $\langle \dot{x}(\tau_1), \vec{n}_1 \rangle > 0$ by assumption, then by the continuity of the vector field of a linear dynamics in $l_0$, there exists an $r > 0$ such that for all $z \in \B_{3 r}(x(\tau_1)) \cap \partial Inv_0$, $\langle \dot{z}, \vec{n}_1 \rangle > 0$ where $\dot{z} := A_0 z + u_0$.
Notice that $\Vert \dot{z} \Vert \le \bar{v}$ by the definition of $\bar{v}$ in (\ref{eq:h:ineq}).
Let $x(t;z)$ denotes the solution at time $t$ of the differential equation $\dot{x}(t) = A_0 x(t)+u_0$ with initial condition $x(0) = z$.
Then for any $z \in \B_{r}(x(\tau_1)) \cap \partial Inv_0$, it is guaranteed that $x(t;z) \in (Inv_0)^C$ for $t \in (0, 2 h)$ for any $h > 0$ satisfying $h < r / \bar{v}$.
This implies that $x(nh) \in (Inv_0)^{c}$ for some $n \in \mathbb{N}$.
Moreover by compactness of $Inv_0$, there exists a $\delta > 0$ such that $\B_{\delta}(x(nh)) \subset (Inv_0)^C$.
\end{proof}

Now suppose that $x(t) \in Inv_0$ for all $0 \le t \le T+\theta$ for some $\theta > 0$.
Then this fact can also be determined.

\begin{lem} \label{lem:theory:trans:noexit}
Suppose $x(t) \in Inv_0$ for all $0 \le t \le T+\theta$ for some $\theta > 0$. 
Then for all small enough $\delta > 0$ and $\gamma > 0$,
\begin{equation}
    \R_{t_f}(\B_{\delta}(x_0), \gamma) \subseteq (Inv_0)^{\circ}.
\end{equation}
where $t_f := \min \{ \tau_1, T \} = T$.
\end{lem}
\begin{proof}
Since $x(t) \in (Inv_0)^{\circ}$ for all $0 \le t \le T$, the result immediately follows from Lemma \ref{lem:theory:cont:2}.
\end{proof}

%=========================================================
\subsection{Over-approximation of Discrete Transition State}  \label{sec:theory:over}

For a given time bound $T$, suppose that the event $\tau_1 < T$ is determined for some $\delta$ and $h$ as shown in Lemma \ref{lem:theory:trans:exit}.
Then, to continue to compute a bounded $\epsilon$-reach set beyond an initial location, we need to determine 
\begin{inparaenum}[(i)]
	\item a new location to which a discrete transition is made from an initial location, and also
	\item an over-approximation of a discrete transition state from which the bounded $\epsilon$-reach set computation can be continued.
\end{inparaenum}
We now show that  these can be determined, if a discrete transition state $x(\tau_1)$ is deterministic and, more importantly, transversal, as defined in Definition \ref{def:transversal}.

\begin{lem} \label{lem:theory:over}
Given $\tau_1 < T$, if $x(\tau_1) \in \partial Inv_0$ satisfies a deterministic and transversal discrete transition condition, then there exists a $\delta > 0$ such that $\B_{2\delta}(x(\tau_1)) \subset (Inv_0 \cup Inv_1)$ for some location $l_1$. Furthermore, there exists a $\Delta > 0$ such that
\begin{enumerate}[(i)]
    \item $x(t) \in (Inv_1)^{\circ}$ for $t \in (\tau_1, \tau_1 + \Delta)$ , and
    \item %moreover, the following hold.
         \begin{equation} \label{eq:lem:lha:6}
                \bigcup_{y \in \mathcal{J}_{0,1}} x(\tau;y) \subset (Inv_1)^{\circ}  \quad \mbox{for } \tau \in (0, \Delta),
         \end{equation}
\end{enumerate}
where $x(\tau;y)$ is the solution at time $\tau$ of an LTI system for the location $l_1$ with an initial state $y$ and $\mathcal{J}_{0,1} := \B_{\delta}(x(\tau_1)) \cap Inv_0 \cap Inv_1$. 
\end{lem}
\begin{proof}
Let $Inv_1, Inv_2$ be invariant sets for some locations $l_1$ and $l_2$ such that $Inv_0 \cap Inv_1 \cap Inv_2 \ne \emptyset$.
Since $x(\tau_1)$ satisfies a deterministic discrete transition condition, if $x(\tau_1) \in Inv_0 \cap Inv_1$, then $x(\tau_1) \notin Inv_0 \cap Inv_2$.
This implies that $x(\tau_1) \not\in Inv_2$.
Then by compactness of $Inv_2$, we know that there exists a $\delta' > 0$ such that $\B_{\delta'}(x(\tau_1)) \cap Inv_2 = \emptyset$. Therefore, we conclude that $\B_{\delta'}(x(\tau_1)) \subset Inv_0 \cup Inv_1$.

Let $\vec{n}_1$ be an outward normal vector of $\partial Inv_0$ at $x(\tau_1)$.
Since $x(\tau_1)$ satisfies a transversal discrete transition condition from the location $l_0$ to the other location $l_1$, we know that there exists a $\delta'' > 0$ such that for all $x(t) \in \B_{\delta''}(x(\tau_1)) \cap Inv_0 \cap Inv_1$, $\langle \dot{x}(t), \vec{n}_1 \rangle > 0$, where $\dot{x}(t)$ is taken as either $A_0 x(t) + u_0$ or as $A_1 x(t) + u_1$, by the continuity of vector fields of the LTI dynamics for $l_0$ and $l_1$.

Let $\delta = \min \lbrace \delta'/2, \delta''/2\rbrace$, and $\Delta := \delta/(2 \bar{v})$ where $\bar{v}$ is as defined in (\ref{eq:h:ineq}).
Then by the definition of $\delta$ and $\bar{v}$, it is clear that (i) and (ii) hold for these choices of $\delta$ and $\Delta$.
\end{proof}

In Lemma \ref{lem:theory:over}, $\mathcal{J}_{0,1}$ is an over-approximation of $x(\tau_1)$ that is determined by taking a $\delta$-ball around $x(\tau_1)$ for suitably small $\delta > 0$, and intersecting it with $Inv_0$ and $Inv_1$. Once such a suitably small $\delta$ is known, then the following lemma shows that it is also possible to determine a $\delta_0$-neighborhood of an initial state $x_0$ such that the reach set at time $\tau_1$ of a DTLHA $\A$ from $\B_{\delta_0}(x_0)$ is contained in $\B_{\delta}(x(\tau_1))$.

\begin{lem} \label{lem:theory:over:cont:1}
Given $\delta$ determined by Lemma \ref{lem:theory:over}, there exists a $\delta_0$ such that
\begin{equation}
       \D_{\tau_1}(\B_{\delta_0}(x_0)) \subseteq \B_{\delta}(x(\tau_1)),
\end{equation}
and $\D_{\tau_1}(\B_{\delta_0}(x_0)) \cap Inv_0 \cap Inv_1$ is an over-approximation of $x(\tau_1)$ determined by $\delta_0$.
\end{lem}
\begin{proof}
This follows from the same argument used in the proof of Lemma \ref{lem:theory:cont:1}, by choosing $\delta_0 = \delta/C$.
\end{proof}

The next lemma shows that $\delta_0$ for $\B_{\delta_0}(x_0)$ can be determined at each discrete transition time $\tau_k$ for $k \ge 1$.

\begin{lem} \label{lem:theory:over:cont:2}
Let $\delta_{k}$ be the radius of a ball centered at $x(\tau_k)$ intersecting only $Inv_{k-1}$ and $Inv_k$, where $\tau_k$ is the $k$-th discrete transition time and $l_k$ is the location after the $k$-th discrete transition. Then for any $x(\tau_k)$ satisfying a deterministic and transversal discrete transition condition, there exists a $\delta_0$ such that
\begin{equation}
       \D_{\tau_k}(\B_{\delta_0}(x_0)) \subseteq \B_{\delta_{k}}(x(\tau_k)) ,
\end{equation}
where $\D_{\tau_k}(\B_{\delta_0}(x_0))$ is the reached states of a given DTLHA $\A$ from $\B_{\delta_0}(x_0)$ at time $\tau_k$.
\end{lem}
\begin{proof}
From the continuity property shown in Lemma \ref{lem:theory:cont:1}, there is a $\delta_{k-1} > 0$ such that $\D_{[0, \tau_k - \tau_{k-1}]}(\B_{\delta_{k-1}}(x(\tau_{k-1})))$ $\subseteq$ $\D_{[0, \tau_k - \tau_{k-1}]}(x(\tau_{k-1}), \delta_{k})$ for a given $\delta_{k}$ where $\D_{[0, \tau_k - \tau_{k-1}]}(x(\tau_{k-1}), \delta_{k})$ denotes a $\delta_k$-approximation of $\D_{[0, \tau_k - \tau_{k-1}]}(x(\tau_{k-1}))$.
Then for this $\delta_{k-1}$, it is clear that $\D_{\tau_k}(\B_{\delta_{k-1}}(x(\tau_{k-1})))$ $\subseteq$ $\B_{\delta_{k}}(x(\tau_k))$.
Using the same argument, we can find $\delta_{k-2}, \delta_{k-3}, \cdots, \delta_1$.
Then from Lemma \ref{lem:theory:over:cont:1}, we know that there exists a $\delta_0 > 0$ such that $\D_{\tau_1}(\B_{\delta_0}(x_0))$ $\subseteq$ $\B_{\delta_1}(x(\tau_1))$.
Since $\D_{\tau_2 - \tau_1}(\B_{\delta_1}(x(\tau_1)))$ $\subseteq$ $\B_{\delta_2}(x(\tau_2))$, we have $\D_{\tau_2}(\B_{\delta_0}(x_0))$ $\subseteq$ $\B_{\delta_2}(x(\tau_2))$.
This relation holds for each $\tau_i$ where $i = 1, 2, \cdots, k$.
Therefore, $\D_{\tau_k}(\B_{\delta_0}(x_0))$ $\subseteq$ $\B_{\delta_k}(x(\tau_k))$. 
\end{proof}

We now present our main result for the bounded $\epsilon$-reachability of a DTLHA.

\begin{thm} \label{thm:theory:thm}
Given $\epsilon > 0$, a time bound $T > 0$, a discrete transition bound $N \in \mathbb{N}$, and a DTLHA $\A$ starting from an initial condition $(l_0, x_0) \in \mathbb{L} \times \mathbb{R}^n$, there exist $\delta > 0$, $\gamma > 0$, and a sampling period $h > 0$ satisfying $h < \gamma/\bar{v}$ such that 
\begin{equation} \label{eq:theory:thm}
	\R_{t_f}(x_0) \subseteq \R_{t_f}(\B_{\delta}(x_0), \gamma) \subseteq \R_{t_f}(x_0, \epsilon),
\end{equation}
where $t_f := \min \{\tau_N, T\}$ and $\tau_N$ is the time at the N-th discrete transition. 
\end{thm}
\begin{proof}
Let $C_i := \max_{0 \le t \le t_f} \{ e^{\Vert A_i \Vert t}\}$ for a location $l_i \in \mathbb{L}$ and $C := \max_{l_i \in \mathbb{L}} \{ C_i \}$.
For a given $\epsilon > 0$, suppose $\delta_k < \epsilon/(4 C)$ at each $\tau_k$ up to $t_f$ where $\delta_k$ is as defined in Lemma \ref{lem:theory:over:cont:2}.
Then, from Lemmas \ref{lem:theory:over}, \ref{lem:theory:over:cont:1}, and \ref{lem:theory:over:cont:2}, we know that there exist a $\delta' > 0$ such that $\D_{\tau_k}(\B_{\delta'}(x_0)) \subseteq \B_{\delta_{k}}(x(\tau_k))$ where $x(t)$ is the execution of a DTLHA $\A$ starting from $x_0$ at time zero.
Furthermore, from Lemmas \ref{lem:theory:trans:exit} and \ref{lem:theory:over}, there also exists $h > 0$ and $\delta'' > 0$ such that 
\begin{inparaenum}[(i)]
	\item $h < \Delta_k$ and
	\item $h$ and $\delta''$ satisfy Lemma \ref{lem:theory:trans:exit}
\end{inparaenum}
at every $\tau_k$ up to $t_f$, where $\Delta_k$ is the $\Delta$ that is defined in Lemma \ref{lem:theory:over} for the $k$-th deterministic and transversal discrete transition.

Let $\hat{\delta} := \min \{ \delta', \delta'' \}$.
Then, with $\hat{\delta}$ and $h$, we can determine every discrete transition event and also construct an over-approximation of the discrete transition state as long as it is deterministic and transversal.
Since $\hat{\delta} \le \delta'$, $\D_{\tau_k}(\B_{\hat{\delta}}(x_0)) \subseteq \B_{\delta_{k}}(x(\tau_k))$ at each $\tau_k$ up to $t_f$.
Thus, for any $\gamma > 0$, 
\[ 
	\D_{[0, \tau_k^{k+1}]}(\D_{\tau_k}(\B_{\hat{\delta}}(x_0)), \gamma) \subseteq \D_{[0, \tau_k^{k+1}]}(\B_{\delta_k}(x_{\tau_k}), \gamma)
\]
where $ \tau_k^{k+1} := \tau_{k+1} - \tau_k$.

Now, we notice that if $\gamma < \epsilon/4$, then from Lemma \ref{lem:theory:cont:2},
\[
	\D_{[0, \tau_k^{k+1}]}(\D_{\tau_k}(\B_{\hat{\delta}}(x_0)), \gamma) \subseteq \D_{[0, \tau_k^{k+1}]}(x(\tau_k), \epsilon),
\]
for each $\tau_k$ up to $t_f$, where the left hand side is a segment of $\R_{t_f}(\B_{\delta}(x_0), \gamma)$ for $[\tau_k, \tau_{k+1}]$, and the right hand side is a segment of $\R_{t_f}(x_0, \epsilon)$ for $[\tau_k, \tau_{k+1}]$ that is defined as $\bigcup_{n = 0}^{N_k -1} \B_{\epsilon}(x(\tau_k + n h))$ where $N_k := \lceil (\tau_{k+1} - \tau_k)/h\rceil$.

Furthermore, if $h < \gamma/\bar{v}$, then from (\ref{eq:h:ineq}) replaced with $\epsilon$ by $\gamma$, it is clear that 
\[
	\D_{[0, \tau_k^{k+1}]}(\D_{\tau_k}(x_0)) \subseteq \D_{[0, \tau_k^{k+1}]}(\D_{\tau_k}(\B_{\hat{\delta}}(x_0)), \gamma),
\]
where the left hand side is a segment of $\R_{t_f}(x_0)$ for $[\tau_k, \tau_{k+1}]$.
Therefore, the result holds.
\end{proof}

%%%%%%%%%%%%%%%%%%%%%%%%%%%%%%%%%%%%%%%%%%%%%%%%%%
\section{Computing a Bounded $\epsilon$-Reach Set of a DTLHA}  \label{sec:cond}

From Theorem \ref{thm:theory:thm}, we know that a set $\R_{t_f}(\B_{\delta}(x_0), \gamma)$, a bounded $\epsilon$-reach set of a DTLHA, can be computed for some $\delta, \gamma$, and $h$. 
In this section, we discuss how to compute $\R_{t_f}(\B_{\delta}(x_0), \gamma)$.
More precisely, we derive a set of conditions, based on the results in Section \ref{sec:theory}, that are needed to correctly detect a deterministic and transversal discrete state transition event and also to determine whether the values for the parameters $\delta, \gamma$, and $h$ are appropriate so as to ensure that $\R_{t_f}(\B_{\delta}(x_0), \gamma)$ is a correct bounded $\epsilon$-reach set.
Furthermore, later in this section, we extend these conditions to incorporate the numerical calculation errors caused by the finite precision numerical calculations capabilities.

%=========================================================
\subsection{Conditions for Bounded $\epsilon$-Reach Set Computation}  \label{sec:cond:exact}

We first note some properties that a set $\R_{t_f}(\B_{\delta}(x_0), \gamma)$ needs to satisfy so that it can be considered as a bounded $\epsilon$-reach set of a DTLHA.

\begin{rem} \label{rem:cond:exact}
Notice that any $\R_{t_f}(\B_{\delta}(x_0), \gamma)$ that can be determined by $\delta, \gamma$, and $h$ in Theorem \ref{thm:theory:thm} for a given $\epsilon > 0$ needs to satisfy the following properties.
\begin{enumerate}[(i)]
	\item $d_H(\R_{t_f}(\B_{\delta}(x_0), \gamma), \R_{t_f}(x_0)) \le \epsilon$, 
	\item  $\R_{t_f}(\B_{\delta}(x_0))) \subset \R_{t_f}(\B_{\delta}(x_0), \gamma)$, and
	\item $d_H(\R_{t_f}(\B_{\delta}(x_0), \gamma), \R_{t_f}(\B_{\delta}(x_0))) \le \gamma$.	
\end{enumerate}
\end{rem}

For given $\delta$ and $h$, the following lemma shows how we can detect a discrete state transition event if there is one.

\begin{lem} \label{lem:cond:exact:trans}
Given a location $l_c$ and a DTLHA $\A$, 
if $\D_{t-h}(\B_{\delta}(x_0)) \subset (Inv_c)^{\circ}$ and $\D_t(\B_{\delta}(x_0)) \subset Inv_c^C$ for some $\delta > 0$ and $h > 0$, where $\B_{\delta}(x_0)$ is a $\delta$-neighborhood of the initial state $x_0$, then there is a discrete transition from the location $l_c$ to some other locations at some time in $(t-h, t)$.
\end{lem}
\begin{proof}
Recall that $\D_t(x_0)$ denotes the reached state of $\A$ at time $t$ from $x_0$.
Then it is clear that $\D_t(x_0) \in \D_t(\B_{\delta}(x_0))$. 
Similarly,  $\D_{t-h}(x_0) \in \D_{t-h}(\B_{\delta}(x_0))$. 
Hence, from the hypothesis, $\D_t(x_0) \in Inv_c^C$ and $\D_{t-h}(x_0) \in (Inv_c)^{\circ}$.
This implies that there exists $\tau \in (t-h ,t)$ such that $\D_{s}(x_0) \in Inv_c^{\circ}$ for $s \in [t-h, \tau)$ and  $\D_{s}(x_0) \in Inv_c^C$ for $s \in (\tau, t]$.
Therefore, there is a discrete transition at some time $\tau \in (t-h, t)$.
\end{proof}

Once a discrete state transition is detected, then, by Lemma \ref{lem:cond:exact:det}, we can check if it is deterministic or not.

\begin{lem} \label{lem:cond:exact:det}
Given an initial state $x_0$ and a DTLHA $\A$,
suppose that there is a discrete transition from a location $l_c$ to some other locations at time $t$, i.e., $\D_{t-h}(\B_{\delta}(x_0)) \subset (Inv_c)^{\circ}$ and $\D_t(\B_{\delta}(x_0)) \subset Inv_c^C$ for some $\delta > 0$ and $h > 0$. 
Then the discrete transition is deterministic if there exists a location $l_n$ such that $l_n \ne l_c$ and $\D_t(\B_{\delta}(x_0)) \subset (Inv_n)^{\circ}$.
\end{lem}%
\begin{proof}
This follows from the definition of a deterministic discrete transition in Definition \ref{def:transversal}.
\end{proof}

We now present conditions to determine the transversality of a discrete state transition; this is more complicated than those in previous two lemmas.
The main idea of the conditions in the following Lemma \ref{lem:cond:exact:transv} is that 
\begin{inparaenum}[(i)]
	\item $\delta$ and $\gamma$ have to be small enough so that every state in an over-approximation of a deterministic and transversal discrete transition state, which can be computed by $\delta$ and $\gamma$, is  also deterministic and transversal, and also
	\item the sampling period $h$ should be small enough so that any reached states right after a discrete state transition can be captured correctly. 
\end{inparaenum}

\begin{lem} \label{lem:cond:exact:transv}
Given $\gamma > 0$ and $h > 0$ satisfying $h < \gamma/\bar{v}$,
suppose that there is a deterministic discrete transition from a location $l_c$ to another location $l_n$ at time $t$, i.e., $\D_{t-h}(\B_{\delta}(x_0)) \subset (Inv_c)^{\circ}$ and $\D_t(\B_{\delta}(x_0)) \subset (Inv_n)^{\circ}$ for some $\delta > 0$ and $h > 0$. 
Then for any $\epsilon > 0$, the discrete transition is transversal if the following conditions hold:
\begin{itemize}
\item[(i)] $h < (dia(\mathcal{J}_{c,n})/2)/(2 \bar{v})$, 
\item[(ii)] $\D_0(\mathcal{J}_{c,n}, dia(\mathcal{J}_{c,n})/2) \subset (Inv_c \cup Inv_n)$, and
\item[(iii)] $\langle \dot{x}_c, \vec{n}_c \rangle \ge \epsilon \land \langle \dot{x}_n, \vec{n}_c \rangle \ge \epsilon, ~~\forall x \in \V(\mathcal{J}_{c,n}')$,
\end{itemize}
where $\mathcal{J}_{c,n} := \D_t(\B_{\delta}(x_0), \gamma) \cap Inv_c \cap Inv_n$, $\mathcal{J}_{c,n}' := \D_0(\mathcal{J}_{c,n},$ $dia(\mathcal{J}_{c,n})/2) \cap Inv_c \cap Inv_n$, $\bar{v}$ is as defined in (\ref{eq:h:ineq}), $\V(\P)$ is a set of vertices of a polyhedron $\P$, $\vec{n}_c$ is an outward normal vector of $\partial Inv_c$, and $\dot{x}_i$ is the vector flow evaluated with respect to the LTI dynamics of location $l_i \in \mathbb{L}$.
\end{lem}
\begin{proof}
Notice that $\D_{t-h}(\B_{\delta}(x_0)) \subset  \D_t(\B_{\delta}(x_0), \gamma)$ since $\gamma$ and $h$ satisfy $h < \gamma/\bar{v}$.
In fact, $\bigcup_{z \in \D_{t-h}(\B_{\delta}(x_0))} x(\tau;z) \subset \D_{t}(\B_{\delta}(x_0), \gamma)$ for $\tau \in [0, h]$ where $x(\tau;z) := e^{A_c \tau} z + \int_{0}^{\tau} e^{A_c s} u_c ds$ under the LTI dynamics of the location $l_c$.
Since $\D_{t-h}(x_0) \in \D_{t-h}(\B_{\delta}(x_0))$ and $\D_t(x_0) \in \D_{t}(\B_{\delta}(x_0))$, $\D_{\tau'}(x_0) \in \mathcal{J}_{c,n}$ for some $\tau' \in (t-h, t)$ where $\D_{\tau'}(x_0)$ is a discrete transition state from $l_c$ to $l_n$ at time $\tau'$.
Thus $\mathcal{J}_{c,n} \ne \emptyset$ (more precisely, $\mathcal{J}_{c,n}^{\circ} \ne \emptyset$) and it is in fact an over-approximation of the deterministic discrete transition state $x_{\tau'} \in Inv_c \cap Inv_n$.

If (ii) and (iii) hold, then it is easy to see that $z'$ satisfies the deterministic and transversal discrete transition condition in Definition \ref{def:transversal} for any $z' \in \mathcal{J}'_{c,n}$.
Now we suppose (i) holds 
and let $x(h;z)$ is the state reached from $z$ at time $h$ under the LTI dynamics of the location $l_n$, then,  for any $z \in \mathcal{J}_{c,n}$, 
\[
	\Vert x(h;z) - z \Vert \le \bar{v}h < dia(\mathcal{J}_{c,n})/2.
\] 

If we now consider the fact that $dia(\mathcal{J}'_{c,n}) \ge 2 \cdot dia(\mathcal{J}_{c,n})$, then it is easy to see that $x(\tau;z) \in Inv_n^{\circ}$ for $\tau \in (0, h)$.
Since $z \in \mathcal{J}_{c,n}$ is arbitrary, we conclude that 
\[
	\D_{\tau}(\mathcal{J}_{c,n}) \in Inv_n^{\circ}
\]
for all $\tau \in (0, h)$.
Thus, the discrete transition state $\D_t(x_0) \in \mathcal{J}_{c,n}$ is transversal and it can be determined through $\mathcal{J}_{c,n}$ with $h$ satisfying (i).
\end{proof}

%=========================================================
\subsection{Finite Precision Basic Calculations}  \label{sec:cond:basic}

Notice that the results in Section \ref{sec:cond:exact} are based on the assumption that the following quantities can be computed exactly:
\begin{itemize}
\item $x(t;x_0) = e^{At}x_0 + \int_0^t e^{As} u ds$.
\item $\H \cap \P$, where $\H$ is a hyperplane and $\P$ is a polyhedron.
\item $hull(\mathcal{V})$, where $hull(\mathcal{V})$ is the convex hull of $\mathcal{V}$ that is a finite set of points in $\mathbb{R}^n$.
\end{itemize}
However, these exact computation assumptions cannot be satisfied in practice and we can only compute each of these with possibly arbitrarily small computation error. 
Therefore, instead of assuming exact computation capabilities for $x(t;x_0)$, $\H \cap \P$, and $hull(\mathcal{V})$, we now assume that the following basic calculation capabilities are available for approximately computing these quantities, and it only these that we can use to compute a bounded $\epsilon$-reach set. 
More precisely, we assume that for given $\mu_c > 0$ and $\mu_h > 0$, 
\begin{itemize}
	\item $a(\mathcal{H} \cap \mathcal{P}, \mu_c)$ and $a(hull(\mathcal{V}), \mu_h)$
\end{itemize}
are available such that $d_H(x, a(x,y)) \le y$, where $a(x, y)$ denotes an approximate computation of $x$, with $y > 0$ as an upper bound on the approximation error.
We also assume that for given $\sigma_e > 0$ and $\sigma_i > 0$,
\begin{itemize}
	\item $a(e^{At}, \sigma_e)$, and $a(\int_0^t e^{A \tau} d\tau, \sigma_i)$
\end{itemize}
are available as an approximate computation of $x(t;x_0)$ such that $\Vert x - a(x,y) \Vert \le y$.
Notice that from these basic calculation capabilities for $x(t;x_0)$, we can compute $a(x(t;x_0), \mu_x)$ with an approximation error  denoted as $\mu_x$, which is upper bounded by a finite value as shown below.

We first note that, for all approximate computations $a(x,y)$ that are used for computing $x(t;x_0)$, we have 
\begin{equation}
x - y \cdot {\bf{1}}_{n \times m} \le a(x, y) \le x + y \cdot {\bf{1}}_{n \times m},
\end{equation}
where $x \in \mathbb{R}^{n \times m}$ and ${\bf{1}}_{n \times m}$ is an $n$ by $m$ matrix whose every element is $1$, and the inequalities hold elementwise.
With this, an upper bound of $\mu_x$ can be derived as follows:
\begin{equation} \nonumber
e^{At} -\sigma_e \cdot {\bf{1}}_{n \times n} \le a( e^{At}, \sigma_e) \le e^{At} + \sigma_e \cdot {\bf{1}}_{n \times n}.
\end{equation}
Similarly, 
\begin{eqnarray} \nonumber
\int_0^t e^{As}ds  - \sigma_i \cdot {\bf{1}}_{n \times n} \le a(\int_0^t e^{As}ds, \sigma_i) \\ \nonumber
\le \int_0^t e^{As}ds + \sigma_i \cdot {\bf{1}}_{n \times n}.
\end{eqnarray}
Hence, we have 
\begin{equation} \nonumber
x(t;x_0) - \delta_x \le a(x(t;x_0), \delta_x)  \le x(t;x_0) + \delta_x,
\end{equation}
where $\delta_x := (\sigma_e \vert x_0 \vert + \sigma_i \vert u \vert) \cdot {\bf{1}}_{n \times 1}$.

Now, we know that $\mu_x$ is upper bounded by the maximum of $\vert \delta_x \vert$ over the continuous state space $\X$ and the control input domain $\U$,
\begin{equation} \label{eq:mux}
	\mu_x \le \max_{x \in \X, u \in \U} \vert \delta_x \vert.
\end{equation}

%=========================================================
\subsection{Conditions for Computation under Finite Precision Calculations}  \label{sec:cond:inexact}

In this section, we extend the results in Section \ref{sec:cond:exact} to derive a set of conditions for a bounded $\epsilon$-reach set computation of the DTLHA under finite precision numerical calculation capabilities.
The following remark is an immediate extension of Remark \ref{rem:cond:exact} in Section \ref{sec:cond:exact}.

In the sequel, for simplicity of notation, we use $\hat{x}$ to denote $a(x,\rho)$ for a given approximation error bound $\rho > 0$.

\begin{rem} \label{rem:cond:inexact}
Let $\hat{\R}_{t_f}(\B_{\delta}(x_0), \gamma)$ be an approximation of $\R_{t_f}(\B_{\delta}(x_0), \gamma)$ that is determined by $\delta, \gamma$, and $h$ in Theorem \ref{thm:theory:thm} and approximate calculations for $x(t;x_0)$, $\H \cap \P$, and $hull(\mathcal{V})$ defined in Section \ref{sec:cond:basic}.
Then, for a given $\epsilon > 0$, it is sufficient for $\hat{\R}_{t_f}(\B_{\delta}(x_0), \gamma)$ to be a bounded $\epsilon$-reach set of a DTLHA $\A$ if the following properties hold.
\begin{enumerate}[(i)]
	\item $d_H(\hat{\R}_{t_f}(\B_{\delta}(x_0), \gamma), \R_{t_f}(x_0)) \le \epsilon$, 
	\item  $\R_{t_f}(\B_{\delta}(x_0))) \subset \hat{\R}_{t_f}(\B_{\delta}(x_0), \gamma)$, and
	\item $d_H(\hat{\R}_{t_f}(\B_{\delta}(x_0), \gamma), \R_{t_f}(\B_{\delta}(x_0))) \le \gamma$.	
\end{enumerate}
\end{rem}

Next, we discuss how the relation between $h$ and $\gamma$ can be modified so as to satisfy (ii) and (iii) in Remark \ref{rem:cond:inexact} when there is numerical calculation error in computing $x(t;x_0)$.

\begin{lem} \label{lem:cond:inexact:isover}
Given a DTLHA $\A$ and its reached state $x(t)$ at time $t$ starting from an initial condition $x(0)$, 
let $\rho > 0$ be an upper bound on the approximation errors such that $\Vert x(t) - \hat{x}(t) \Vert \le \rho$.
If a given sampling period $h$ satisfies $h < (\gamma - \rho)/\bar{v}$ for a given $\gamma$ satisfying $\gamma > \rho$,  where $\bar{v}$ is as defined in (\ref{eq:h:ineq}), then the following property holds at any location $l_i \in \mathbb{L}$ of $\A$:
\begin{equation}
	x(t+\tau) \subset \B_{\gamma}(\hat{x}(t)), ~~ \forall \tau \in [0, h],
\end{equation}
where $x(t+\tau) = e^{A_i \tau} x(t) + \int_0^{\tau} e^{A_i s} u_i ds$.
\end{lem}
\begin{proof}
Since $\Vert x(t) - \hat{x}(t) \Vert \le \rho$, $x(t) \in \B_{\rho}(\hat{x}(t))$.
Moreover, from (\ref{eq:h:ineq:1}), we know that for any $x(t) \in \X$,
\begin{eqnarray} \nonumber
	\max_{\tau \in [0,h]} \Vert x(t+\tau) - x(t) \Vert &\le& \max_{\tau \in [0,h]}  \int_t^{t+\tau} \Vert \dot{x}(s) \Vert ds \\ \nonumber
	 &\le& \bar{v} h. 
\end{eqnarray}
Hence, if $h < (\gamma - \rho)/\bar{v}$, then, for any $x(t) \in \X$, 
\[
	\max_{\tau \in [0,h]} \Vert x(t+\tau) - x(t) \Vert < \gamma - \rho. 
\] 
This means that $x(t+\tau) \in \B_{\gamma - \rho}(x(t))$ for $\tau \in [0, h]$.
Therefore, for $\tau \in [0, h]$, 
\begin{eqnarray} \nonumber
	\Vert \hat{x}(t) - x(t+\tau)\Vert &\le& \Vert \hat{x}(t) - x(t) \Vert + \Vert x(t) -x(t+\tau) \Vert \\ \nonumber
	&\le& \rho +  (\gamma -\rho).
\end{eqnarray}
Thus $\Vert \hat{x}(t) - x(t+\tau)\Vert  \le \gamma$.
\end{proof}

Notice that Lemma \ref{lem:cond:inexact:isover} says that if $h < (\gamma - \rho)/\bar{v}$ for a given $\rho > 0$, then a $\gamma$-neighborhood of a sampled state is indeed an over-approximation of a trajectory over the time interval $h$.
We now extend the result in Lemma \ref{lem:cond:inexact:isover} to the case where we need to compute a $\gamma$-approximation of a polyhedron.

\begin{lem} \label{lem:cond:inexact:over}
Given a DTLHA $\A$ and its reached states $\D_t(\B_{\delta}(x_0))$ at some time $t$ from initial states in $\B_{\delta}(x_0)$, 
let $\rho > 0$ be an upper bound on the approximation errors such that $d_H(\D_t(\B_{\delta}(x_0)), \hat{\D}_t(\B_{\delta}(x_0))) \le \rho$.
If a given sampling period $h$ satisfies the following inequality 
\begin{equation} \label{eq:lem:cond:inexact:over}
h < \frac{\gamma - \rho}{\bar{v}} ,
\end{equation}
then, for a given $\gamma$ satisfying $\gamma > \rho$, 
\begin{equation}
	\D_{t+ \tau}(\B_{\delta}(x_0)) \subset \hat{\D}_t(\B_{\delta}(x_0), \gamma), \quad \forall \tau \in [0, h], 
\end{equation}
 where $\hat{\D}_{t}(\B_{\delta}(x_0), \gamma)$ is a $\gamma$-approximation of $\hat{\D}_{t}(\B_{\delta}(x_0))$ that is  constructed as the convex hull of the set of extreme points of a polyhedral $\gamma$-neighborhood of all vertices of $\hat{\D}_{t}(\B_{\delta}(x_0))$ and $\bar{v}$ is as defined in (\ref{eq:h:ineq}).
\end{lem}
\begin{proof}
Let $\V$ and $\hat{\V}$ be the set of extreme points of $\D_t(\B_{\delta}(x_0))$ and $\hat{\D}_t(\B_{\delta}(x_0))$, respectively. 
Since $d_H(\D_t(\B_{\delta}(x_0)), \hat{\D}_t(\B_{\delta}(x_0))) \le \rho$ and $\gamma > \rho$, it is clear that $\D_{t}(\B_{\delta}(x_0)) \subset \hat{\D}_t(\B_{\delta}(x_0), \gamma)$.
From Lemma \ref{lem:cond:inexact:isover}, we know that for each $x(t) \in \V$, $x(t+\tau) \subset \B_{\gamma}(\hat{x})$ for all $\tau \in [0, h]$ where $\hat{x} \in \hat{\V}$ corresponding to $x(t)$.
Let $\V_{t+\tau}$ be the set of extreme points of $\D_{t+ \tau}(\B_{\delta}(x_0))$. 
Then $\V_{t+\tau} \subset \hat{\D}_t(\B_{\delta}(x_0), \gamma)$ for all $\tau \in [0, h]$ since 
\begin{inparaenum}[(i)]
	\item for each $x(t) \in \V$, $x(t+\tau) \subset \B_{\gamma}(\hat{x})$ for all $\tau \in [0, h]$ and 
	\item from the construction of $\hat{\D}_t(\B_{\delta}(x_0), \gamma)$, $\B_{\gamma}(\hat{x}) \subset \hat{\D}_t(\B_{\delta}(x_0), \gamma)$ for each $\hat{x} \in \hat{\V}$.
\end{inparaenum}
Therefore, the convex hull of $\V_{t+\tau}$, which is $\D_{t+ \tau}(\B_{\delta}(x_0))$, has to be contained in $\hat{\D}_t(\B_{\delta}(x_0), \gamma)$ for all $\tau \in [0,h]$ since $\hat{\D}_t(\B_{\delta}(x_0), \gamma)$ is convex and $\V_{t+\tau} \subset \hat{\D}_t(\B_{\delta}(x_0), \gamma)$ for all $\tau \in [0, h]$.
\end{proof}

For (i) in Remark \ref{rem:cond:inexact}, Lemma \ref{lem:cond:inexact:eps} below shows that the diameter of a set $\hat{\D}_{t}(\B_{\delta}(x_0), \gamma)$ has to be smaller than a given $\epsilon > 0$.

\begin{lem} \label{lem:cond:inexact:eps}
Given $\epsilon > 0$, $\delta > 0$, $\gamma >0$,  $\rho > 0$, and a DTLHA $\A$, 
suppose a given sampling period $h > 0$ satisfies the inequality (\ref{eq:lem:cond:inexact:over}).
Then $\D_{[t, t+h]}(x_0) \subset \hat{\D}_{t}(\B_{\delta}(x_0), \gamma)$ and $d_H(\hat{\D}_{t}(\B_{\delta}(x_0), \gamma), \D_{[t, t+h]}(x_0)) \le \epsilon$, if the following hold:
\begin{equation} \label{eq:lem:cond:inexact:eps}
	dia(\hat{\D}_t(\B_{\delta}(x_0), \gamma)) \le \epsilon,
\end{equation}
where $\D_{[t, t+h]}(x_0)$ is the set of reached states of $\A$ starting from $x_0$ during the time interval $[t, t+h]$. 
\end{lem}
\begin{proof}
Since $h$ satisfies (\ref{eq:lem:cond:inexact:over}), it is trivial to see that  $\D_{[t, t+h]}(x_0) \subset \hat{\D}_{t}(\B_{\delta}(x_0), \gamma)$ holds from Lemma \ref{lem:cond:inexact:over}.
Moreover, if (\ref{eq:lem:cond:inexact:eps}) is also true, then for any $z \in \D_{[t, t+h]}(x_0)$, $\max_{y \in \hat{\D}_{t}(\B_{\delta}(x_0), \gamma)} \Vert y - z \Vert \le \epsilon$ since $z \in \hat{\D}_{t}(\B_{\delta}(x_0), \gamma)$.
Therefore, it is clear that $d_H(\hat{\D}_{t}(\B_{\delta}(x_0), \gamma), \D_{[t, t+h]}(x_0)) \le \epsilon$ if (\ref{eq:lem:cond:inexact:over}) and (\ref{eq:lem:cond:inexact:eps}) hold.
\end{proof}

Now we can extend the results of Lemmas \ref{lem:cond:exact:trans}, \ref{lem:cond:exact:det}, and \ref{lem:cond:exact:transv} to incorporate a numerical calculation error $\rho > 0$.

\begin{lem} \label{lem:cond:inexact:istrans}
Given $\rho > 0$,  a location $l_c$, and $\hat{\D}_t(\B_{\delta}(x_0))$ at time $t$,
if
\begin{enumerate}[(i)]
\item $\hat{\D}_{t-h}(\B_{\delta}(x_0),$ $\rho) \subset (Inv_c)^{\circ}$, and
\item $\hat{\D}_t(\B_{\delta}(x_0), \rho) \subset Inv_c^C$ 
\end{enumerate} 
for some $\delta > 0$ and $h > 0$, then there is a discrete transition from the location $l_c$ to some other locations.
\end{lem}
\begin{proof}
Notice that $d_H(\D_t(\B_{\delta}(x_0)), \hat{\D}_t(\B_{\delta}(x_0))) \le \rho$, which implies $\D_t(\B_{\delta}(x_0)) \subset \hat{\D}_t(\B_{\delta}(x_0), \rho)$.
Similarly, $\D_{t-h}(\B_{\delta}(x_0))$ $\subset \hat{\D}_{t-h}(\B_{\delta}(x_0), \rho)$.
Hence if (i) and (ii) hold, then it is clear that $\D_t(\B_{\delta}(x_0)) \subset Inv_c^C$ and $\D_{t-h}(\B_{\delta}(x_0)) \subset (Inv_c)^{\circ}$.
Then the result follows from Lemma \ref{lem:cond:exact:trans}. 
\end{proof}

\begin{lem} \label{lem:cond:inexact:det}
Given $\rho > 0$, a location $l_c$, and $\hat{\D}_t(\B_{\delta}(x_0))$ at time $t$, suppose that a discrete transition from a location $l_c$ to some other locations is determined as in Lemma \ref{lem:cond:inexact:istrans}.
Then the discrete transition is a deterministic discrete transition from $l_c$ to $l_n$ if there exists a location $l_n$ such that $l_n \ne l_c$ and $\hat{\D}_t(\B_{\delta}(x_0), \rho) \subset (Inv_n)^{\circ}$.
\end{lem}
\begin{proof}
Notice that $\D_{t-h}(\B_{\delta}(x_0)) \subset (Inv_c)^{\circ}$ from the result in Lemma \ref{lem:cond:inexact:istrans}.
Since $\D_t(\B_{\delta}(x_0)) \subset \hat{\D}_t(\B_{\delta}(x_0), \rho)$, if  $\hat{\D}_t(\B_{\delta}(x_0), \rho) \subset (Inv_n)^{\circ}$, then  $\D_t(\B_{\delta}(x_0)) \subset (Inv_n)^{\circ}$.
Thus by Lemma \ref{lem:cond:exact:det}, the conclusion holds.
\end{proof}

\begin{lem} \label{lem:cond:inexact:transv}
Given $\rho > 0$, $\gamma > 0$ and $h > 0$ satisfying (\ref{eq:lem:cond:inexact:over}), suppose that a deterministic discrete transition from a location $l_c$ to another location $l_n$ is determined as in Lemma \ref{lem:cond:inexact:istrans} and Lemma \ref{lem:cond:inexact:det}, i.e., $\hat{\D}_{t-h}(\B_{\delta}(x_0), \rho) \subset (Inv_c)^{\circ}$ and $\hat{\D}_{t}(\B_{\delta}(x_0), \rho) \subset (Inv_n)^{\circ}$.
Then, for any $\epsilon > 0$, the discrete transition is transversal if the following conditions hold:
\begin{itemize}
\item[(i)] $h < (dia(\hat{\mathcal{J}}_{c,n})/2)/(2 \bar{v})$, 
 \item[(ii)] $\D_0(\hat{\mathcal{J}}_{c,n}, dia(\hat{\mathcal{J}}_{c,n})/2 + \rho) \subset (Inv_c \cup Inv_n)$, and 
 \item[(iii)] $\langle \dot{x}_c, \vec{n}_c \rangle \ge \epsilon \land \langle \dot{x}_n, \vec{n}_c \rangle \ge \epsilon,~~\forall x \in \V(\hat{\mathcal{J}}_{c,n}')$,
 \end{itemize}
 where $\hat{\mathcal{J}}_{c,n} := \hat{\D}_{t}(\B_{\delta}(x_0), \gamma + \rho) \cap Inv_c \cap Inv_n$, $\hat{\mathcal{J}}_{c,n}' := \D_0(\hat{\mathcal{J}}_{c,n},$ $dia(\hat{\mathcal{J}}_{c,n})/2 + \rho) \cap Inv_c \cap Inv_n$, $ \hat{\D}_{t}(\B_{\delta}(x_0), \gamma + \rho)$ is a $(\gamma + \rho)$-approximation of $ \hat{\D}_{t}(\B_{\delta}(x_0))$, and $\dot{x}_i$ and $\vec{n}_c$ are as defined in Lemma \ref{lem:cond:exact:transv}. 
\end{lem}
\begin{proof}
Notice that $\D_{t}(\B_{\delta}(x_0), \gamma) \subset \hat{\D}_{t}(\B_{\delta}(x_0), \gamma + \rho)$.
Then, by the definition of $\mathcal{J}_{c,n}$ given in Lemma \ref{lem:cond:exact:transv} and $\hat{\mathcal{J}}_{c,n}$, we know $\mathcal{J}_{c,n} \subset \hat{\mathcal{J}}_{c,n}$.
Hence, $\hat{\mathcal{J}}_{c,n} \ne \emptyset$ since $\mathcal{J}_{c,n} \ne \emptyset$ by the construction of  $\mathcal{J}_{c,n}$.
Now if (i) holds, then it is easy to see that $\D_{\tau}(\hat{\mathcal{J}}_{c,n}) \subset \D_0(\hat{\mathcal{J}}_{c,n}, dia(\hat{\mathcal{J}}_{c,n})/2)$ for $\tau \in (0, h)$.
Moreover, (ii) and (iii) imply that $\D_{\tau}(\hat{\mathcal{J}}_{c,n})$ is in fact contained in $Inv_n^{\circ}$ for $\tau \in (0, h)$.
\end{proof}

%%%%%%%%%%%%%%%%%%%%%%%%%%%%%%%%%%%%%%%%%%%%%%%%%%
\section{Architecture and Algorithm for Bounded $\epsilon$-Reach Set Computation of a DTLHA}  \label{sec:design}

We are now in a position to propose an algorithm for bounded $\epsilon$-reach set computation of a DTLHA.
Before proving its correctness, we first describe its architecture.

For \emph{flexibility}, we decouple the higher levels of the algorithm, called \emph{Policy}, from the component, called \emph{Mechanisms}, where specific steps of calculations are performed through some numerical routines.
The proposed architecture of the algorithm, shown in Fig. \ref{fig:arch}, consists of roughly five different components \emph{Policy}, \emph{Mechanism}, \emph{System Description}, \emph{Data}, and \emph{Numerics}.
A more detailed explanation of each of these modules is given below.

\begin{figure}
\begin{center}
\includegraphics[width = 9cm]{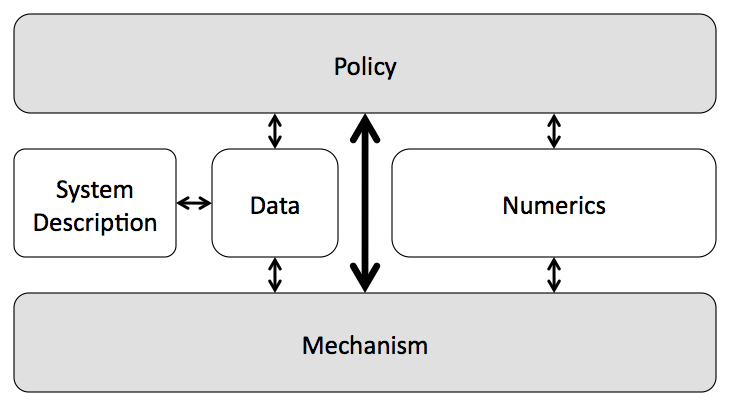} 	%% for 'pdftex' typeset
\caption{An architecture for bounded $\epsilon$-reach set computation.}
\label{fig:arch}
\end{center}
\end{figure}

The System Description contains all information describing a problem of a bounded $\epsilon$-reach set computation of a DTLHA.
This consists of $\X$, the domain of continuous state space, a DTLHA $\A$, and an initial condition $(l_0, x_0) \in \mathbb{L} \times \X$. 
Also, an upper bound $T \in \mathbb{R}^+$ on terminal time, an upper bound $N \in \mathbb{N}$ on the total number of discrete transitions, and an approximation parameter $\epsilon > 0$, are described.
A bounded $\epsilon$-reach set of a DTLHA $\A$ is computed in the Mechanism component based on a given set of numerical calculation algorithms in Numerics, as well as a given Policy, which captures some of the higher level choices of the algorithm's outer loops. 
In the Data component, all computation data that is relevant to a computed bounded $\epsilon$-reach set, generated on-the-fly in the Mechanism part, are stored. 
Each of the functions in Numerics is in fact an implementation of some numerical computation algorithms.  
As an example, $e^{At}$ can be computed in many different ways as shown in \cite{moler:78} and each of the different algorithms can compute the value with a certain accuracy. 
Here we assume that a set of such numerical computation algorithms for basic calculations are given\footnote{In this way, we decouple the low-level numerical calculations from our bounded $\epsilon$-reach set algorithm. This is the reason why the Numerics component is represented separately from the Mechanism component.} 
and the corresponding approximation error bounds, i.e., $\sigma_e, \sigma_i, \mu_c$, and $\mu_h$, are known a priori.
The Policy component represents a user-defined rules that choose appropriate values of the parameters, especially $\delta > 0$, $\gamma > 0$, and $h > 0$, which are needed to continue to compute a bounded $\epsilon$-reach set of a DTLHA, when a bounded $\epsilon$-reach set algorithm in Mechanism fails to determine some events or to satisfy some required properties, during its computation.
The Mechanism component represents the core of the bounded $\epsilon$-reach set algorithm based on the theoretical results in Section \ref{sec:theory} and \ref{sec:cond}, and is detailed in Section \ref{sec:design:algo}. 
Given values for parameters $\delta > 0$, $\gamma > 0$, and $h > 0$, it computes a bounded $\epsilon$-reach set of a DTLHA $\A$ until it either successfully finishes its computation or cannot make further progress, which happens when some required conditions or properties are not met.
Notice that,  as stated in Section \ref{sec:cond}, there are a set of conditions and properties that a computed set needs to satisfy to be a correct bounded $\epsilon$-reach set.
If the algorithm fails to resolve a computation, then it returns to Policy indicating the problems so that a user-defined rule in Policy can choose another set of values for the parameters to resolve the problems.
Every computation result is stored in the Data component to be possibly used later in Policy and Mechanism.

%=========================================================
\subsection{Core Algorithm for Bounded $\epsilon$-Reach Set of a DTLHA}  \label{sec:design:algo}

An algorithm to compute a bounded $\epsilon$-reach set of a DTLHA is proposed and shown in Algorithm \ref{alg:main}.
Let $k$ indicate a computation step of the algorithm from which the proposed algorithm starts its bounded $\epsilon$-reach set computation.
All computation history up to the $(k-1)$-th computation step is stored as data, called {\tt Reached}, in Data part.
Then, given an input $(k, \delta_k, \gamma_k, h_k)$ from Policy, the algorithm first retrieves the computation data at the $(k-1)$-th computation step from {\tt Reached} and starts its $k$-th computation step using this data.
As shown in Algorithm \ref{alg:main}, the algorithm continues its computation until it either
\begin{inparaenum}[(i)]
	\item returns {\tt done} when it successfully finished to compute a bounded $\epsilon$-reach set or
	\item returns {\tt error} when it encounters some erroneous situations during the execution of a function, called {\tt Post()}.
\end{inparaenum}
If the algorithm returns an {\tt error}, it also indicates the cause of the {\tt error} so that a user-defined rule in Policy can choose appropriate values for the input parameters.

%-------------------------
% Algorithm
%-------------------------
\begin{algorithm}  %[htbp]
\SetAlgoLined 
%\LinesNumbered
\BlankLine
\KwIn{$k, \delta_k, \gamma_k, h_k$ from Policy} 
\BlankLine
\Compute $\mu_x$ from $(\sigma_e, \sigma_i)$\\
\BlankLine
\While{\True} {
\Get data at $(k-1)$-th step from \Reached \\
\If{$\delta_k \ne \delta_{k-1}$} {
	\Compute $\hat{\D}_{t_{k-1}}(\B_{\delta_k}(x_0))$\\ %and $\rho_{k-1}'$
	\Update $\rho_{k-1}$ %$\rho_k \leftarrow \max \{\rho_{k-1}, \rho_{k-1}' \} + \mu_x$
}
%
%\lIf{$h_k \ge (\gamma_k - \rho_k)/\bar{v}$}{\Return \Err}\\
%
$t_k \leftarrow t_{k-1}+h_k$\\
\Call  \Post  \\
\Store $k$-th computation data into \Reached \\
$k \leftarrow k+1$\\
\lIf{$(t_k \ge T) \lor (\Jump \ge N)$}{\Return \Done}
}
\BlankLine
\caption{An algorithm for bounded $\epsilon$-reach set computation of a DTLHA.}
\label{alg:main}
\end{algorithm}
%-------------------------

%-------------------------
% Algorithm
%-------------------------
\begin{algorithm} %[htbp]
\SetAlgoLined 
%\LinesNumbered
\BlankLine
\KwIn{$h_k, \gamma_k, l_{k-1}, \rho_{k-1}, \hat{\D}_{t_{k-1}}(\B_{\delta_k}(x_0))$} 
\BlankLine
\Compute $\hat{\D}_{t_k}(\B_{\delta_k}(x_0))$ from $\hat{\D}_{t_{k-1}}(\B_{\delta_k}(x_0))$\\
\Compute $ \hat{\D}_{t_k}(\B_{\delta_k}(x_0),\gamma_k)$ from $\hat{\D}_{t_k}(\B_{\delta_k}(x_0))$\\
%\Update  $\rho_k \leftarrow \rho_{k-1} + (\mu_x + \mu_h)$\\
\Update  $\rho_k \leftarrow \rho_{k-1} + \mu_x$\\
\BlankLine
\lIf{$h_k \ge (\gamma_k - \rho_k)/\bar{v}$}{\Return \Err}\\
\lIf{$dia(\hat{\D}_{t_k}(\B_{\delta_k}(x_0), \gamma_k)) \ge \epsilon$} {\Return \Err} \\
\uIf{$\hat{\D}_{t_k}(\B_{\delta_k}(x_0)) \cap Inv(l_{k-1}) = \emptyset$}{
	\If{$\hat{\D}_{t_{k-1}}(\B_{\delta_k}(x_0)) \subset Inv(l_{k-1})^{\circ}$} {
		\uIf{$\Deterministic \land \Transversal$}{
			\Update $\hat{\D}_{t_k}(\B_{\delta_k}(x_0))$ and $\hat{\D}_{t_k}(\B_{\delta_k}(x_0), \gamma_k)$
			\Update $\rho_k \leftarrow \rho_k + \mu_x + \mu_c + \mu_h$\\
			\Update $l_k$\\
			 \Jump $\leftarrow$ \Jump + 1
		}
		\lElse{\Return \Err}
	}
}
\uElseIf{$\hat{\D}_{t_k}(\B_{\delta_k}(x_0)) \not\subset Inv(l_{k-1})^{\circ}$}{\Return \Err}
\lElse{$l_k \leftarrow l_{k-1}$}
\BlankLine
\caption{A function {\tt Post()}.}
\label{alg:post}
\end{algorithm}
%-------------------------

In the proposed algorithm in {\tt Post()}, $\hat{\D}_{t_k}(\B_{\delta_k}(x_0))$ is computed from $\hat{\D}_{t_{k-1}}(\B_{\delta_k}(x_0))$ as follows:

Given a polyhedron $\hat{\D}_{t_{k-1}}(\B_{\delta_k}(x_0))$, we first compute the set of the vertices of $\hat{\D}_{t_{k-1}}(\B_{\delta_k}(x_0))$ that is denoted as $\V$.
Then for each $v_i \in \V$, we compute 
\[
	v_i(h_k) := e^{A_k h_k} v_i + \int_0^{h_k} e^{A_k s} u_k ds
\] 
where $A_k$ and $u_k$ are given by the linear dynamics of a location $l_k$ on which the linear image of $\hat{\D}_{t_{k-1}}(\B_{\delta_k}(x_0))$ is computed at the $k$-the computation step in Algorithm \ref{alg:main}. 
If we let $\V_h := \{v_i(h_k) : v_i \in \V \}$, then we can compute $\hat{\D}_{t_k}(\B_{\delta_k}(x_0))$ as follows: 
\[
	\hat{\D}_{t_k}(\B_{\delta_k}(x_0)) := hull(\V_h)
\]
where $hull(\V_h)$ is the convex hull of $\V_h$. 

Once we have $\hat{\D}_{t_k}(\B_{\delta_k}(x_0))$, we compute $\hat{\D}_{t_k}(\B_{\delta_k}(x_0), \gamma_k)$ in the following way.
To compute $\hat{\D}_{t_k}(\B_{\delta_k}(x_0), \gamma_k)$ for a given $\gamma_k$, we first construct a hypercubic $\gamma_k$-neighborhood of $v_i(h_k)$ for each $v_i(h_k) \in \V_h$.
Let $\B_{\gamma_k}(v_i(h_k))$ be such a $\gamma_k$ hypercubic neighborhood of $v_i(h_k)$ and $\V_{h}^{\gamma}$ be the set of vertices of $\B_{\gamma_k}(v_i(h_k))$ for all $v_i(h_k) \in \V_h$.
Then we can compute $\hat{\D}_{t_k}(\B_{\delta_k}(x_0), \gamma_k)$ as follows:
\begin{equation} \label{sec:design:algo:hull}
	\hat{\D}_{t_k}(\B_{\delta_k}(x_0), \gamma_k) := hull(\V_h^{\gamma}).
\end{equation}

This process of polyhedral image computation under a linear dynamics is illustrated in Fig. \ref{fig:post:linear}.
We now show that $\hat{\D}_{t_k}(\B_{\delta_k}(x_0), \gamma_k)$ that is computed as in (\ref{sec:design:algo:hull}) is indeed a $\gamma_k$-approximation of $\hat{\D}_{t_k}(\B_{\delta_k}(x_0))$ for a given $\gamma_k$.

\begin{figure} [htbp]
\begin{center}
	\includegraphics[width=8cm]{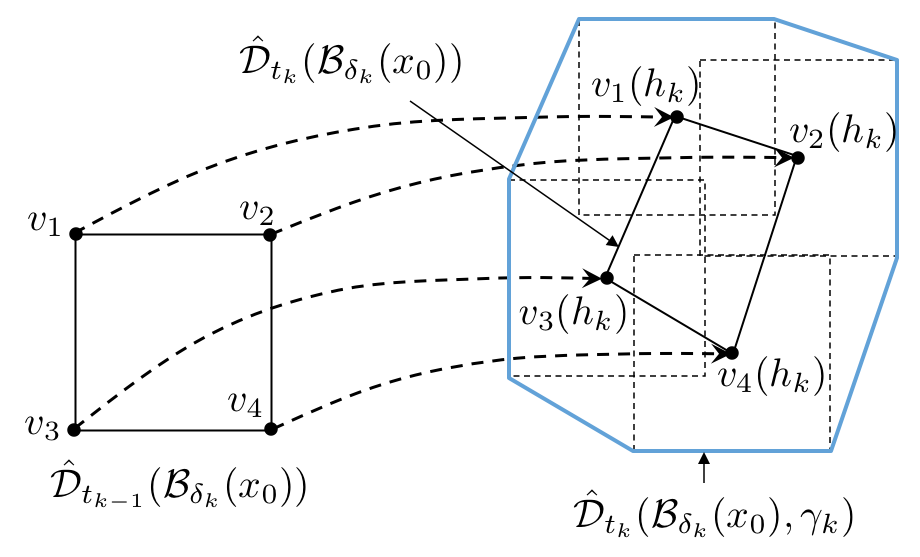}
\caption{The image computation under a linear dynamics.}
\label{fig:post:linear}
\end{center}
\end{figure}

\begin{lem}
Let $\H$ be the convex hull of $\V_h^{\gamma}$. 
Then $\H$ is exactly the closed $\gamma$-neighborhood of the convex hull of $\V_h$.
\end{lem}
\begin{proof}
Suppose $\bar{w} \in \H$ and $\bar{w} \not\in hull(\V_h)$.
Then $\bar{w} = \lambda \bar{y}_1 + (1-\lambda)\bar{y}_2$ for some $\bar{y}_1$ and $\bar{y}_2$ such that $\Vert \bar{y}_1 - v_1 \Vert \le \gamma$ and $\Vert \bar{y}_2 - v_2 \Vert \le \gamma$ for some $v_1, v_2 \in \V_h$ and $0 \le \lambda \le 1$.
Then there exists $v = \lambda v_1 + (1-\lambda)v_2 \in hull(\V_h)$ such that
\begin{eqnarray} \nonumber
\Vert \bar{w} -v \Vert &=& \Vert \lambda(\bar{y}_1 - v_1) + (1-\lambda)(\bar{y}_2 - v_2) \Vert \nonumber \\
&\le& \lambda \Vert \bar{y}_1 - v_1 \Vert + (1-\lambda) \Vert \bar{y}_2 - v_2 \Vert  \nonumber \\
&\le& \gamma. \nonumber
\end{eqnarray}
Thus $\bar{w}$ is in the $\gamma$-neighborhood of the convex hull of $\V_h$.

For the converse, consider $\bar{z}$ in the $\gamma$-neighborhood of the convex hull of $\V_h$.
Then for some $\lambda_i \ge 0$, $\sum_i \lambda_i = 1$, $\Vert \bar{z} - \sum_i \lambda_i v_i(h) \Vert \le \gamma$, where $v_i(h) \in \V_h$.
Let $s := \bar{z} - \sum_i \lambda_i v_i(h)$.
Now $\bar{z} = \sum_i \lambda_i(v_i(h) + s)$.
So $\bar{z}$ is in the convex hull of $\lbrace v_i(h) + s \rbrace$.
However each $v_i(h) + s \in \B_{\gamma}(v_i(h))$.
Hence each $v_i(h) + s$ is in the convex hull of the vertices of $\B_{\gamma}(v_i(h))$ which is $\H$.
Thus $\bar{z}$ is in $\H$.
\end{proof}

Notice that the first update of $\rho_k$ in {\tt Post()} is due to the computation of $\hat{\D}_{t_k}(\B_{\delta_k}(x_0))$ from $\hat{\D}_{t_{k-1}}(\B_{\delta_k}(x_0))$ over the time interval $h_k$ under the linear dynamics of $l_{k-1}$.
The second update after a deterministic and transversal discrete transition is due to a series of computations from $\hat{\D}_{t_k}(\B_{\delta_k}(x_0))$ that is used to determine such a discrete transition to a new $\hat{\D}_{t_k}(\B_{\delta_k}(x_0))$ that represents a reached states at time $t_k$ right after a deterministic and transversal discrete transition.
As described in Lemma \ref{lem:cond:inexact:transv}, the steps involved during this discrete transition are to compute
\begin{inparaenum}[(i)]
	\item $\hat{\mathcal{J}}_{c,n}$ from $\hat{\D}_{t_k}(\B_{\delta_k}(x_0))$ and
	\item $\D_{h_k}(\hat{\mathcal{J}}_{c,n})$ from $\hat{\mathcal{J}}_{c,n}$. 
\end{inparaenum}
Notice that (i) requires an intersection between a hyperplane and a polyhedron as well as a convex hull computation.
Moreover, for (ii), we need to compute a polyhedral image under the linear dynamics of a new location that is determined in {\tt Post()}.
Recall that we have derived a set of conditions in Lemmas \ref{lem:cond:inexact:istrans}, \ref{lem:cond:inexact:det}, and \ref{lem:cond:inexact:transv} to determine a deterministic and transversal discrete transition event.
These conditions are used in {\tt Post()} to determine such an event.
Furthermore, we also use conditions derived in Lemmas \ref{lem:cond:inexact:over} and \ref{lem:cond:inexact:eps}, to ensure that a set $\hat{\R}_{t_f}(\B_{\delta}(x_0), \gamma)$, which can be constructed as a collection of $\hat{\D}_{t_k}(\B_{\delta_k}(x_0), \gamma_k)$ as shown in the following theorem, satisfies the properties given in Remark \ref{rem:cond:inexact}.

Now, we present our main result for the problem of computing a bounded $\epsilon$-reach set of a DTLHA.

\begin{thm}
Given input $(\X, \A, l_0, x_0, T,$ $N, \epsilon)$ for a problem to compute a bounded $\epsilon$-reach set of a DTLHA $\A$, 
if Algorithm \ref{alg:main} returns {\tt done}, then a bounded $\epsilon$-reach set of a DTLHA $\A$ defined over the continuous state domain $\X$ starting from an initial condition $(l_0, x_0) \in \mathbb{L} \times \mathbb{R}^n$, denoted as $\hat{\R}_{t_f}(x_0, \epsilon)$, is the following:
\begin{equation}
\hat{\R}_{t_f}(x_0, \epsilon) := \bigcup_{k = 1}^{K}  \hat{\D}_{t_k}(\B_{\delta_k}(x_0),\gamma_k),
\end{equation}
for some $K \in \mathbb{N}$ where $t_f := \min \{T, \tau_N \}$ and $\tau_N$ is the time at the $N$-th discrete transition. 
\end{thm}
\begin{proof}
For each $k \le K$, 
\begin{inparaenum}
\item[(i)] $\gamma_k, h_k, \rho_k$ satisfies Lemma \ref{lem:cond:inexact:over}, and 
\item[(ii)] $\hat{\D}_{t_k}(\B_{\delta_k}(x_0),\gamma_k)$ satisfies Lemma \ref{lem:cond:inexact:eps}. 
\end{inparaenum}
Hence $\hat{\D}_{t_k}(\B_{\delta_k}(x_0), \gamma_k)$ is guaranteed to satisfy 
$\D_{[t, t+h]}(x_0) \subset \hat{\D}_{t_k}(\B_{\delta_k}(x_0), \gamma_k)$ and $d_H(\hat{\D}_{t_k}(\B_{\delta_k}(x_0), \gamma_k), \D_{[t, t+h]}(x_0)) \le \epsilon$.
Furthermore, if a deterministic and transversal discrete transition is detected at the $k$-th step by $\hat{\D}_{t_k}(\B_{\delta_k}(x_0))$, then
\begin{inparaenum}
	\item[(iii)] by Lemmas \ref{lem:cond:inexact:istrans}, \ref{lem:cond:inexact:det}, and \ref{lem:cond:inexact:transv}, there is in fact a deterministic and transversal discrete transition in $(t_{k-1}, t_k)$. 
\end{inparaenum}
This implies that a deterministic and transversal discrete transition event is correctly determined. 
Finally, if the proposed algorithm returns {\tt done}, then this implies that 
\begin{inparaenum}
	\item[(iv)] either $t_k \ge T$ or ${\tt jump} \ge N$. 
\end{inparaenum}
Hence, $t_f$ is $\min \{ T, \tau_N\}$.
Therefore, $\R_{t_f}$ is a bounded $\epsilon$-reach set of $\A$ from $x_0$ by (i), (ii), (iii), and (iv).
\end{proof}

%%%%%%%%%%%%%%%%%%%%%%%%%%%%%%%%%%%%%%%%%%%%%%%%%%
\section{Optimization and Implementation of the Proposed Algorithm}  \label{sec:imp}

A prototype software tool has been implemented, based on the architecture and the algorithm proposed in Section \ref{sec:design}, to demonstrate the idea of a bounded $\epsilon$-reach set computation. 
In our implementation, we use the Multi-Parametric Toolbox \cite{mpt} for polyhedral operations and also use some built-in Matlab functions for other calculations.

Notice that the size of the $\hat{\D}_t(\B_{\delta}(x_0))$ right after a discrete transition increases roughly by the amount $\gamma$ through the computation of $\hat{\mathcal{J}}_{c,n}$. 
This can potentially affect the capability to determine a discrete transition event.
Hence, we determine a smaller value of $\gamma$ to construct a tighter over-approximation of a discrete transition state. 
Suppose that a discrete transition from a location $l_i$ to some other location $l_j$ has already been determined by the proposed algorithm for given $h > 0$, $\hat{\D}_{t-h}(\B_{\delta}(x_0), \rho)$, and $\hat{\D}_{t}(\B_{\delta}(x_0), \rho)$ at some time $t$.
Then the procedure for construction of a tight over-approximation of a discrete transition state $x(\tau_k)$ for some $\tau_k \in (t-h, t)$ can be improved shown in Algorithm \ref{alg:opt}.

%-------------------------
% Algorithm
%-------------------------
\begin{algorithm} %[htbp]
\SetAlgoLined 
%\LinesNumbered
\BlankLine
1. Partition the time interval $[t-h, t]$ into a finite sequence of $\{ I_m \}_{m=1}^M$ for some $M \in \mathbb{N}$, where
\[ I_m : = [t-h +(m-1) \cdot \Delta h, t-h + m \cdot \Delta h] \]
for some $\Delta h \ll h$. \\
\BlankLine
2. Find a time $\tau := t-h + m \cdot \Delta h \in (t-h, t)$ such that 
	\begin{itemize}
		\item $volume(\P_{\tau}^i) > volume(\P_{\tau}^j)$ and
		\item $volume(\P_{\tau + \Delta h}^i) < volume(\P_{\tau + \Delta h}^j)$, 
	\end{itemize}
	where $\P_t^k := Inv_k \cap  \hat{\D}_{t}(\B_{\delta}(x_0), \rho)$. \\
\BlankLine
3. Construct $\hat{\D}_{\tau+\Delta h}(\B_{\delta}(x_0), \gamma' + \rho)$ where $\gamma' > \Delta h \cdot \bar{v}$. \\
\BlankLine
4. Compute an over-approximate discrete transition state
\[ \hat{\mathcal{J}}_{i,j} := \hat{\D}_{\tau+\Delta h}(\B_{\delta}(x_0), \gamma' + \rho) \cap Inv_i \cap Inv_j. \]
%
%\BlankLine
\caption{A procedure to compute a tight over-approximation of discrete transition state.}
\label{alg:opt}
\end{algorithm}
%-------------------------

%=========================================================
\subsection{An Example of Bounded $\epsilon$-Reach Set Computation}  \label{sec:imp:example}

As an example to evaluate the proposed algorithm for a bounded $\epsilon$-reach set computation of a DTLHA $\A$, we consider an LHA $\A := (\mathbb{L}, Inv, A, u, \xrightarrow{G})$ over a  continuous state space $\X := [-8, 8] \times [-8, 8] \subset \mathbb{R}^2$ where
\begin{enumerate}[(i)]
	\item $\mathbb{L} = \{Up, Down,$ $Left, Right \}$,
	\item $A(l)$ and $u(l)$ for each location $l \in \mathbb{L}$ are defined as shown in Table \ref{tab:ex:Au},  
	\item The invariant set for each location $l \in \mathbb{L}$, $Inv(l)$, is defined as shown in Fig. \ref{fig:eg:rho}, and
	\item $\xrightarrow{G}$ holds at the intersection between invariant sets of different locations. 
\end{enumerate}

Notice that all the LTI dynamics defined in the given LHA $\A$ are asymptotically stable.
Moreover, from the the vector fields determined by $A(l)$ and $u(l)$ for each $l \in \mathbb{L}$, every discrete transition which occurs along the boundary of the invariant set between different locations is deterministic and transversal. 
Hence the given LHA $\A$ is in fact a DTLHA.

\begin{table}[htdp]
\caption{$A(l)$ and $u(l)$ for each $l \in \mathbb{L}$ of $\A$}
\begin{center}
{\small
\begin{tabular}{c|p{0.15\textwidth}|c}
\toprule
$l$ & \centering{$A(l)$} & $u(l)$ \\ 
\midrule
$UP$ & \centering{$\begin{pmatrix}  -0.2 & -1 \\ 3 & -0.2 \end{pmatrix}$} & $\begin{pmatrix}  0.1 \\ 0.1 \end{pmatrix}$ \\ 
$DOWN$ & \centering{$\begin{pmatrix}  -0.2 & -1 \\ 3 & -0.2 \end{pmatrix}$} & $\begin{pmatrix}  -0.2 \\ -0.2 \end{pmatrix}$ \\ 
$LEFT$ & \centering{$\begin{pmatrix}  -0.2 & -3 \\ 1 & -0.2 \end{pmatrix}$} & $\begin{pmatrix} 0.15 \\ 0.15 \end{pmatrix}$ \\ 
$RIGHT$ & \centering{$\begin{pmatrix}  -0.2 & -3 \\ 1 & -0.2 \end{pmatrix}$} & $\begin{pmatrix} 0.3 \\ 0.3 \end{pmatrix}$ \\ 
\bottomrule
\end{tabular}
}
\end{center}
\label{tab:ex:Au}
\end{table}%

The bounded $\epsilon$-reach set computation problem is specified by $(\A, l_0, x_0, T, N, \epsilon)$ where $l_0 = Up$, $x_0 = (2.5, 6)^T$, $T = 20$ sec., $N = 10$, and $\epsilon = 0.5$.

In this example, we also assume that numerical calculation algorithms are available for basic calculations defined in Section \ref{sec:cond:basic} such that $a(e^{At}, \rho)$, $a(\int_{0}^{t} e^{A \tau} d \tau, \rho)$, $a(\H \cap \P, \rho)$, and $a(hull(\V), \rho)$ where $\rho$ is specified as $10^{-15}$.

\begin{figure}
\begin{center}
	\includegraphics[width=8.5cm]{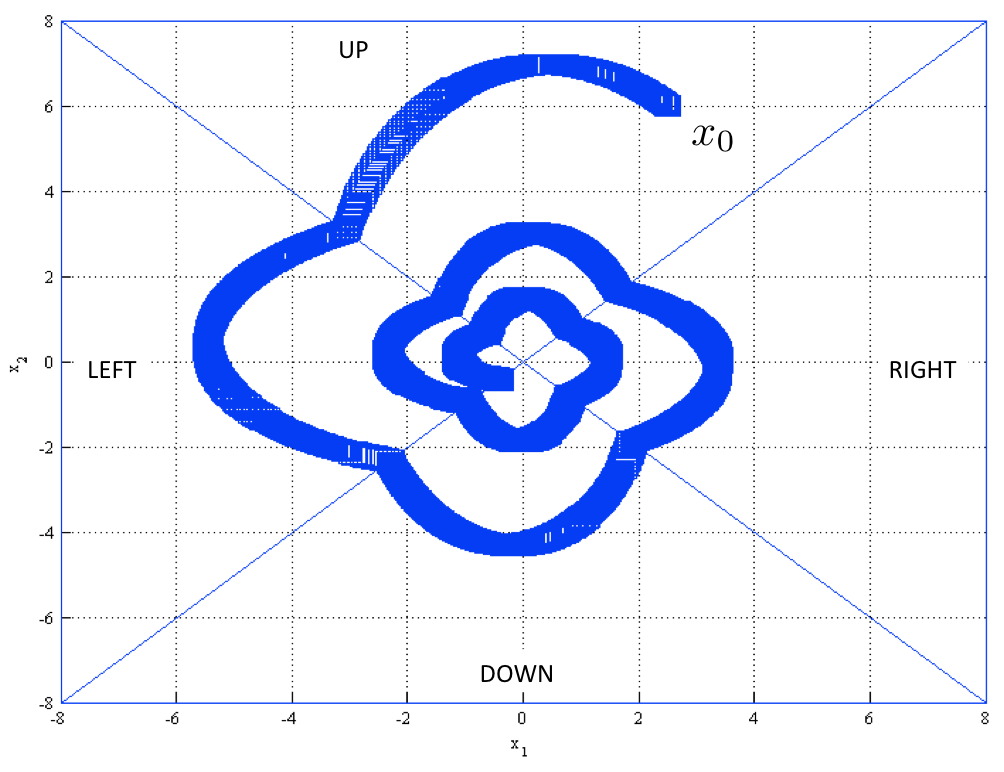}
\caption{A bounded $\epsilon$-reach set of a DTLHA $\A$ starting from $(Up, [2.5, 6]^T)$.}	
\label{fig:eg:rho}
\end{center}
\end{figure}

A policy that is used to choose values for $(k, \delta_k, \gamma_k, h_k)$ is as follows:
\begin{enumerate}[(i)]
	\item $k$ is chosen in non-decreasing manner, 
	\item $\delta_k := 10^{-5}$ to define a fixed sufficiently small $\B_{\delta}(x_0)$,
	\item $\gamma_k := (\epsilon - dia(\hat{\D}_{t_k}(\B_{\delta_k}(x_0), \rho_k)))/2$, and
	\item $h_k := (\gamma_k/2)/ \bar{v}$ where $\bar{v}$ is as defined in (\ref{eq:h:ineq}).
\end{enumerate}
Notice that (i) means that whenever the proposed $\epsilon$-reach set algorithm fails to continue its computation at the $k$-th computation step, then the policy decides to restart the computation from the $k$-th step with different values of the other parameters.
Recall that $\rho_k$ denotes the approximation error of $\hat{\D}_{t_k}(\B_{\delta_k}(x_0))$ when the algorithm computes $\D_{t_k}(\B_{\delta_k}(x_0))$ at time $t_k$.
As shown in (iii), for a given $\epsilon$, the policy chooses the largest value of $\gamma_k$ at each computation step.
The equation for $\gamma_k$ given in (iii) can easily be derived by considering 
\begin{equation} \nonumber
	dia(\hat{\D}_{t_k}(\B_{\delta_k}(x_0), \gamma_k + \rho_k)) \le dia(\hat{\D}_{t_k}(\B_{\delta_k}(x_0), \rho_k)) + 2 \gamma_k.
\end{equation}
If we upper bound the right hand side by $\epsilon$, then we have (iii).

Fig. \ref{fig:eg:rho} shows the computation result.
As shown in Fig. \ref{fig:eg:rho}, a bounded $\epsilon$-reach set is successfully computed.
In this example, the algorithm terminates at the computation step $k = 2259$ right after the algorithm makes the tenth discrete transition from locations $Left$ to $Down$ at the time $t = 12.1415$ sec. and {\tt jump $= 10$}.
For given $\rho := 10^{-15}$, the accumulated numerical calculation error $\rho_k$ for $\D_{t_k}(\B_{\delta_k}(x_0))$ at this termination time is $2.5638 \times 10^{-11}$.

%%%%%%%%%%%%%%%%%%%%%%%%%%%%%%%%%%%%%%%%%%%%%%%%%%
\section{Conclusion} \label{sec:con}
We have defined a special class of hybrid automata, called Deterministic and Transversal Linear Hybrid Automata (DTLHA), for which we can address the problem of bounded $\epsilon$-reach set computation starting from an initial state.
For this class, we can also incorporate the impact of numerical calculation errors caused by finite precision numerical computation.

It is of importance to determine more  general and useful models of hybrid systems that permit computational verification of safety properties. Hybrid linear systems that incorporate linear models widely employed in control systems are a natural candidate around which to build such a theory of verification and validation.

%%%%%%%%%%%%%%%%%%%%%%%%%%%%%%%%%%%%%%%%%%%%%%%%%%
% if have a single appendix:
%\appendix[Proof of the Zonklar Equations]
% or
%\appendix  % for no appendix heading
% do not use \section anymore after \appendix, only \section*
% is possibly needed

% use appendices with more than one appendix
% then use \section to start each appendix
% you must declare a \section before using any
% \subsection or using \label (\appendices by itself
% starts a section numbered zero.)
%

%\appendices
%\section{Proof of the First Zonklar Equation}
%Appendix one text goes here.
%
%% you can choose not to have a title for an appendix
%% if you want by leaving the argument blank
%\section{}
%Appendix two text goes here.

%%%%%%%%%%%%%%%%%%%%%%%%%%%%%%%%%%%%%%%%%%%%%%%%%%%%%%

%% use section* for acknowledgement
%\section*{Acknowledgment}
%The authors would like to thank...

% Can use something like this to put references on a page
% by themselves when using endfloat and the captionsoff option.
\ifCLASSOPTIONcaptionsoff
  \newpage
\fi

% trigger a \newpage just before the given reference
% number - used to balance the columns on the last page
% adjust value as needed - may need to be readjusted if
% the document is modified later
%\IEEEtriggeratref{8}
% The "triggered" command can be changed if desired:
%\IEEEtriggercmd{\enlargethispage{-5in}}

% references section

% can use a bibliography generated by BibTeX as a .bbl file
% BibTeX documentation can be easily obtained at:
% http://www.ctan.org/tex-archive/biblio/bibtex/contrib/doc/
% The IEEEtran BibTeX style support page is at:
% http://www.michaelshell.org/tex/ieeetran/bibtex/
%\bibliographystyle{IEEEtran}
% argument is your BibTeX string definitions and bibliography database(s)
%\bibliography{IEEEabrv,../bib/paper}
%
% <OR> manually copy in the resultant .bbl file
% set second argument of \begin to the number of references
% (used to reserve space for the reference number labels box)
%\begin{thebibliography}{1}
%
%\bibitem{IEEEhowto:kopka}
%H.~Kopka and P.~W. Daly, \emph{A Guide to \LaTeX}, 3rd~ed.\hskip 1em plus
%  0.5em minus 0.4em\relax Harlow, England: Addison-Wesley, 1999.
%
%\end{thebibliography}

%\nocite{frehse:08,asarin:00,chutinan:03,henzinger:97,henzinger:95,alex:02,uppaal}
\bibliographystyle{IEEEtran}
\bibliography{Bounded_epsilon-Reach_Set_DTLHA}

% biography section
% 
% If you have an EPS/PDF photo (graphicx package needed) extra braces are
% needed around the contents of the optional argument to biography to prevent
% the LaTeX parser from getting confused when it sees the complicated
% \includegraphics command within an optional argument. (You could create
% your own custom macro containing the \includegraphics command to make things
% simpler here.)
%\begin{biography}[{\includegraphics[width=1in,height=1.25in,clip,keepaspectratio]{mshell}}]{Michael Shell}
% or if you just want to reserve a space for a photo:

%\begin{IEEEbiography}[{\includegraphics[width=1in,height=1.25in,clip,keepaspectratio]{kdkim.png}}]{Kyoung-Dae kim}
%\begin{IEEEbiography}{Kyoung-Dae kim}
%Biography text here.
%\end{IEEEbiography}
%
%
%\begin{IEEEbiography}{Sayan Mitra}
%Biography text here.
%\end{IEEEbiography}
%
%
%\begin{IEEEbiography}{P. R. Kumar}
%Biography text here.
%\end{IEEEbiography}

%% if you will not have a photo at all:
%\begin{IEEEbiographynophoto}{Kyoung-Dae Kim}
%Biography text here.
%\end{IEEEbiographynophoto}
%
%% insert where needed to balance the two columns on the last page with
%% biographies
%%\newpage
%
%\begin{IEEEbiographynophoto}{Sayan Mitra}
%Biography text here.
%\end{IEEEbiographynophoto}
%
%
%\begin{IEEEbiographynophoto}{P.R. Kumar}
%Biography text here.
%\end{IEEEbiographynophoto}

% You can push biographies down or up by placing
% a \vfill before or after them. The appropriate
% use of \vfill depends on what kind of text is
% on the last page and whether or not the columns
% are being equalized.

\vfill

% Can be used to pull up biographies so that the bottom of the last one
% is flush with the other column.
%\enlargethispage{-5in}

% that's all folks
\end{document}